\documentclass[sigconf,nonacm]{acmart}

\AtBeginDocument{%
  }

\setcopyright{acmcopyright}
\copyrightyear{2023}
\acmYear{2023}
\acmDOI{10.1145/3583668.xxxxxxx}

\acmConference[PODC '23]{ACM Symposium on Principles of Distributed Computing}{June 19--23, 2023}{Orlando, FL, USA}
\acmBooktitle{ACM Symposium on Principles of Distributed Computing (PODC '23), June 19--23, 2023, Orlando, FL, USA}
\acmPrice{15.00}
\acmISBN{979-8-4007-0121-4/23/06}

\usepackage[draft]{todonotes}

\usepackage{thm-restate}

\usepackage{thmtools}
\usepackage{mathtools}

\usepackage{dsfont} 

\usepackage{subcaption}

\newcommand*{\namedFunction}[3]{\ensuremath{#1_{#2}\left(#3\right)}}
\newcommand*{\namedFunctionCurly}[3]{\ensuremath{#1_{#2}\left\{#3\right\}}}

\newcommand*{\dist}[3][]{{\ensuremath{\text{dist}}}_{{#1}}({#2,#3})}
\newcommand*{\Dist}[2][]{{\ensuremath{\text{dist}}}_{{#1}}({#2})}
\newcommand*{\buy}[2][]{\text{buy}_{{#1}}({#2})}
\newcommand*{\Cost}[2][]{\text{cost}_{{#1}}({#2})}

\newcommand*{\layer}[2][]{\namedFunction{\ell}{#1}{#2}}
\newcommand*{\minset}[1]{\namedFunctionCurly{\min}{}{#1}}
\newcommand*{\maxset}[1]{\namedFunctionCurly{\max}{}{#1}}
\newcommand*{\autoset}[1]{\namedFunctionCurly{}{}{#1}}
\newcommand*{\depth}[1]{\ensuremath{\textrm{depth}}{}({#1})}

\newcommand*{\neig}[3][]{\ensuremath{\textrm{Neigh}_{#1}^{#2}\left(#3\right)}}

\newcommand*{\poa}[2][]{\ensuremath{\rho_{#1}\left(#2\right)}}
\newcommand*{\card}[1]{\ensuremath{\left|{#1}\right|}}
\newcommand*{\floor}[1]{\ensuremath{\left\lfloor{#1}\right\rfloor}}
\newcommand*{\ceil}[1]{\ensuremath{\left\lceil{#1}\right\rceil}}

    \newcommand*{\N}{\ensuremath{\mathds{N}}}
    
    \newcommand*{\R}{\ensuremath{\mathds{R}}}

\newcommand{\opt}{\ensuremath{\textrm{OPT}}}

\usepackage[nameinlink]{cleveref}

\begin{document}

\title{The Impact of Cooperation in Bilateral Network Creation}

\author{Tobias Friedrich}
\email{tobias.friedrich@hpi.de}
\orcid{0000-0003-0076-6308}
\affiliation{%
  \institution{Hasso Plattner Institute\\University of Potsdam}
  \city{Potsdam}
  \country{Germany}
}

\author{Hans Gawendowicz}
\email{hans.gawendowicz@hpi.de}
\orcid{0000-0001-8394-4469}
\affiliation{%
	\institution{Hasso Plattner Institute\\University of Potsdam}
	\city{Potsdam}
	\country{Germany}
}

\author{Pascal Lenzner}
\email{pascal.lenzner@hpi.de}
\orcid{0000-0002-3010-1019}
\affiliation{%
	\institution{Hasso Plattner Institute\\University of Potsdam}
	\city{Potsdam}
	\country{Germany}
}

\author{Arthur Zahn}
\email{arthur.zahn@student.hpi.de}
\orcid{0009-0007-2904-7098}
\affiliation{%
	\institution{Hasso Plattner Institute\\University of Potsdam}
	\city{Potsdam}
	\country{Germany}
}

\renewcommand{\shortauthors}{Friedrich et al.}

\begin{abstract}
  Many real-world networks, like the Internet or social networks, are not the result of central design but instead the outcome of the interaction of local agents that selfishly optimize their individual utility. The well-known Network Creation Game by~Fabrikant, Luthra, Maneva, Papadimitriou, and Shenker~\cite{fabrikant2003network} models this. There, agents corresponding to network nodes buy incident edges towards other agents for a price of~\(\alpha > 0\) and simultaneously try to minimize their buying cost and their total hop distance.
Since in many real-world networks, e.g., social networks, consent from both sides is required to establish and maintain a connection, \citet{corbo2005price} proposed a bilateral version of the Network Creation Game, in which mutual consent and payment are required in order to create edges. It is known that this cooperative version has a significantly higher Price of Anarchy compared to the unilateral version. On the first glance this is counter-intuitive, since cooperation should help to avoid socially bad states. However, in the bilateral version only a very restrictive form of cooperation is considered.

We investigate this trade-off between the amount of cooperation and the Price of Anarchy by analyzing the bilateral version with respect to various degrees of cooperation among the agents. With this, we provide insights into what kind of cooperation is needed to ensure that socially good networks are created.
As a first step in this direction, we focus on tree networks and present a collection of asymptotically tight bounds on the Price of Anarchy that precisely map the impact of cooperation. Most strikingly, we find that weak forms of cooperation already yield a significantly improved Price of Anarchy. In particular, the cooperation of coalitions of size 3 is enough to achieve constant bounds. Moreover, for general networks we show that enhanced cooperation yields close to optimal networks for a wide range of edge prices. Along the way, we disprove an old conjecture by \citet{corbo2005price}.
\end{abstract}

\ccsdesc[500]{Theory of computation~Network formation}
\ccsdesc[500]{Theory of computation~Quality of equilibria}
\ccsdesc[500]{Theory of computation~Social networks}
\ccsdesc[500]{Theory of computation~Algorithmic game theory}

\keywords{Network Creation Games, Cooperation, Price of Anarchy}

\maketitle

\section{Introduction}
Many real-world problems are related to networks or connections between entities. Given this, research on networks is concerned with the creation of efficient networks with regard to various objective functions. Traditionally, such networks are created by a centralized algorithm, like in the cases of Minimum Spanning Trees \cite{graham1985history}, Topology Control Problems~\cite{rajaraman2002topology}, and Network Design Problems~\cite{magnanti1984network}.
Such central orchestration is a good model if the whole network belongs to one real-world entity, i.e., a firm designing its internal communication network, who governs the network and centrally pays for its infrastructure. However, most large real-world networks instead emerged from the interaction of multiple entities, each with their own goals, who each have control over a local part of the network. For example, the Internet is a network of networks where each subnetwork is centrally controlled by an Internet service provider. Hence, to understand the formation of such networks, game-theoretic agent-based models are needed.
For such models, one of the prime questions is to understand the impact of the agents' selfishness on the overall quality of the created networks.
This impact is typically measured by the Price of Anarchy~(PoA)~\cite{KP09}.

The Network Creation Game (NCG) by Fabrikant, Luthra, Maneva, Papadimitriou, Shenker~\cite{fabrikant2003network} is one of the most prominent game-theoretic models for the formation of networks. There, the agents correspond to nodes in a network and each agent can unilaterally build incident edges to other nodes for an edge price of~\(\alpha>0\). All agents simultaneously try to minimize (1) their total hop distance to everyone else and (2) the cost they incur for building connections. As both incentives are in direct conflict, the challenge is to analyze how this tension is resolved in the equilibria of the game.
The NCG found widespread appeal
and many variants have recently been studied.
One of the earliest variants is the Bilateral Network Creation Game (BNCG) by Corbo and Parkes~\cite{corbo2005price}, which models that not all real-world settings allow unilateral edge formation. For instance, social networks require mutual consent and both sides to invest their time and effort in order to maintain a connection. Consequently, the BNCG demands that both incident agents pay for an edge to establish it. From this arises a need for coordination. Thus, Corbo and Parkes~\cite{corbo2005price} did not analyze the standard Pure Nash Equilibrium (NE) as solution concept, but instead focused on the well-known concept of Pairwise Stability (PS)~\cite{jackson1996strategic}, in which two agents are allowed to cooperate to form a mutual connection. 

Interestingly, despite an abundance of literature on the NCG, its bilateral variant remains widely unexplored. This is even more astonishing given the known results on the PoA for both models: While the PoA with respect to NE of the NCG is known to be constant for most ranges of $\alpha$, the PoA with respect to PS of the BNCG was proven to be high~\cite{corbo2005price,demaine2009price}. Hence, the required cooperation for establishing edges leads to socially worse equilibrium states.
But cooperation among the agents should be beneficial and should reduce the social cost of equilibrium states. Thus, the problem with the analysis of the BNCG seems to be the drastically limited amount of cooperation allowed by Pairwise Stability. This directly gives rise to a very natural question: How much cooperation among the agents is actually needed to ensure a low Price of Anarchy?   

We answer this question by investigating the impact of different amounts of cooperation on the quality of the resulting equilibrium states, i.e., the social cost of the created networks.
We establish the positive result that allowing slightly more cooperation than allowed by PS already has a significant impact on the PoA. 

\subsection{Model and Notation}
\label{sec:model_and_notation}
For~\(i, j\in\N\), with~\(i\leq j\), we denote the set~\(\autoset{k\in\N \mid i\leq k\leq j}\) as \([i..j]\) and the set~\([1..i]\) as~\([i]\). Whenever we use the logarithm, the base is always~\(2\) unless specified otherwise.

\textbf{Graphs:} We model networks as \emph{graphs}, consisting of nodes and undirected edges, defined as a pair~\(G=(V_G, E_G)\), where~\(V_G\) is the set of nodes and~\(E_G\) the set of edges.
The number of nodes~\(\card{V_G}\) is~\(n_G > 0\). We omit the subscript if the graph is clear from the context. Each edge~\(e\in E_G\) is a subset of~\(V_G\) and consists of the two distinct nodes it connects. For two nodes~\(u,v\in V_G\), we use~\(uv\) for \(\{u,v\}\) and say that~\(u\) and~\(v\) are \emph{neighbors} in~\(G\) if~\(uv\in E_G\).

The \emph{distance}~\(\dist[G]{u}{v}\) between two nodes~\(u,v\in V\) in a graph~\(G\) is defined as the number of edges on a shortest path from~\(u\) to~\(v\) in $G$. Since our graphs are undirected, we have \(\dist{u}{v} = \dist{v}{u}\). Furthermore, we assume~\(\dist{u}{u} = 0\).
If no path exists, we consider the distance to be extremely large. For technical reasons, we define the constant~\(M\in\N\) and say~\(\dist{u}{v} = M\) in such cases. We elaborate on this constant~\(M\) below.
The \emph{total distance} of a node~\(u\in V\) to multiple other nodes~\(V'\subseteq V\) is defined by~\(\dist{u}{V'} = \sum_{v\in V'}\dist{u}{v}\). In the special case~\(V=V'\), we use the shorthand~\(\Dist{u} = \dist{u}{V}\) and call this the total distance cost of node~\(u\).
We define the \emph{extended neighborhood} of a node~\(u\in V\) as follows: for~\(i\in\N\), we refer to the set of all nodes of distance at most~\(i\) from~\(u\) as the set \(\neig[G]{\leq i}{u} = \autoset{v\in V\mid \dist{u}{v}\leq i}\). Analogously, we define~\(\neig{=i}{u}\) and note that~\(\neig{=1}{u}\) denotes the \emph{neighborhood} of~\(u\).

For a graph~\(G\) and~\(u,v\in V\), the graph~\(G-uv\) refers to~\(G\) without edge~\(uv\), i.e.,~\((V, E\setminus \autoset{uv})\). By contrast, if we consider adding~\(uv\) to~\(G\), we denote the resulting graph as~\(G+uv\).

\textbf{(Bilateral) Network Creation Game:}
The Network Creation Game consists of~\(n\) selfish \emph{agents} who are building a network among each other. These agents correspond to the node set~\(V\) of the resulting graph, so we will use the terms agent and node interchangeably. Each agent~\(u\in V\) has a \emph{strategy}~\(S_u\subseteq V\setminus\autoset{u}\), which specifies towards which other nodes the agent~\(u\) wants to create edges. The individual strategies of all agents are combined into a \emph{strategy vector}~\(S\), which encapsulates a state of the game.
The \emph{created graph}~\(G\) is derived from this strategy vector. We consider two versions, the \emph{Unilateral Network Creation Game (NCG)}~\cite{fabrikant2003network} and the \emph{Bilateral Network Creation Game (BNCG)}~\cite{corbo2005price}.
In the NCG, for two nodes~\(u,v\in V\), the edge~\(uv\) is part of~\(G\) exactly if~\(u\in S_v\) or~\(v\in S_u\). In contrast, in the BNCG, we have~\(uv\in E\) if and only if~\(u\in S_v\) and~\(v\in S_u\), i.e., both agents need to agree on the creation of the edge. In this work, we will mostly focus on the BNCG.

Building edges is not free in the (B)NCG. Instead, there is a parameter~\(\alpha\in\R_{>0}\) which determines the buying cost of each edge. An agent~\(u\in V\) incurs the buying cost~\(\buy[S]{u} = \alpha\card{S_u}\), so she has to pay for each target node in her strategy. Note that in the BNCG an agent~\(u\) has to pay for~\(v\in S_u\) even if~\(u\notin S_v\) and thus~\(uv\notin E\). Similarly, edges might be paid for twice in the NCG if~\(u\in S_v\) and~\(v\in S_u\). However, such inefficiencies will not arise in equilibria and hence we will ignore them. Given this, we have a bijection between strategy vectors and created graphs in the BNCG. Thus, we can abstract away from the underlying strategy vector~\(S\) and directly consider the created graph~\(G\).

Each agent~$u$ aims to minimize her \emph{total cost} in graph $G$, denoted by~$\Cost[G]{u}$ and defined as the sum of her buying cost and her total hop distance to all other agents, i.e.,
$$\Cost[G]{u} = \buy[G]{u} + \Dist[G]{u} = \alpha|S_u| + \sum_{v\in V_G}\dist[G]{u}{v}.$$

Remember that we define the distance between two disconnected nodes to be~\(M\). We set $M$ to some value larger than \(\alpha n^3\) to enforce that, for~\(i\in [n]\), an agent should prefer any graph where she can reach~\(i\) agents over any graphs where she can reach at most~\(i-1\) agents, but given the same reachability she should prefer to minimize her buying and distance cost.

\textbf{Solution Concepts:}
As all agents are selfish, the created graphs are the result of decentralized decisions instead of central coordination. 
When analyzing the properties of such games, we are especially interested in \emph{equilibria}, which describe strategy vectors that are stable against specific types of strategy changes. Non-equilibrium states might not persist,
because the agents would defect from such states in order to decrease their cost.
We consider different solution concepts, where each of them is characterized by the types of improving strategy changes it is stable against. We say that a strategy change is \emph{improving} for some agent if the change yields strictly lower cost for this agent.

The most prominent solution concept for the NCG is the \emph{Pure Nash Equilibrium (NE)}.
A strategy vector~\(S\) is in NE if and only if no agent~\(u\in V\) can strictly decrease her cost by changing her strategy~\(S_u\). However, the NE is not suited for the BNCG, since a node~\(u\in V\) may delete edges unilaterally by removing them from her strategy, but she cannot add new edges by changing only her own strategy.
Hence, a meaningful solution concept for the BNCG requires that multiple agents are able to coordinate to change their strategies in an atomic step. 
We consider the following solution concepts, presented in order of increasing amount of cooperation. Newly introduced concepts are marked with a "$\star$", concepts adapted from the unilateral version are marked with "$\dagger$".
\begin{itemize}
 \item \textbf{Remove Equilibrium$^\star$(RE):} A strategy vector~\(S\) is a \emph{Remove Equilibrium (RE)} if no agent \(u\in V\) can improve by removing a single node from her strategy~\(S_u\).
 \item \textbf{Bilateral Add Equilibrium$^\dagger$(BAE):} A strategy vector~\(S\) is a \emph{Bilateral Add Equilibrium (BAE)}\footnote{Chauhan, Lenzner, Melnichenko, and Molitor~\cite{chauhan2017selfish} considered Add-Only Equilibria for a unilateral variant of the NCG. Also there, only edge additions are allowed. A bilateral version was also considered by Bullinger, Lenzner, and Melnichenko~\cite{BLM22}.} if there are no~\(u,v\in V\) such that both improve by adding~\(u\) to~\(S_v\) and~\(v\) to~\(S_u\).  
 \item \textbf{Pairwise Stability (PS)~\cite{jackson1996strategic}:} A strategy vector~\(S\) is \emph{pairwise stable (PS)}, if it is in Remove Equilibrium and in Bilateral Add Equilibrium.
 \item \textbf{Bilateral Swap Equilibrium$^\dagger$(BSwE):} A strategy vector~\(S\) is a \emph{Bilateral Swap Equilibrium (BSwE)}\footnote{Mihalák and Schlegel~\cite{mihalak2012asymmetric} study Asymmetric Swap Equilibria for the NCG. The BSwE extends this.} if there are no nodes \(u,v,w\in V\) with~\(u\in S_v\),~\(v\in S_u\) and~\(u\notin S_w\) or \(w\notin S_u\) such that~\(u\) and~\(w\) can improve by replacing~\(v\) by~\(w\) in~\(S_u\) and adding~\(u\) to~\(S_w\). 
 \item \textbf{Bilateral Greedy Equilibrium$^\dagger$(BGE):} A strategy vector~\(S\) is in \emph{Bilateral Greedy Equilibrium (BGE)}\footnote{For the NCG, Lenzner~\cite{lenzner2012greedy} defined the Greedy Equilibrium (GE), where no agent can improve by unilaterally adding, removing, or swapping a single incident edge. The BGE is the natural extension.} if it is pairwise stable and in Bilateral Swap Equilibrium.
 \item \textbf{Bilateral Neighborhood Equilibrium$^\star$(BNE):} A strategy vector~\(S\) is in \emph{Bilateral Neighborhood Equilibrium (BNE)} if there is no agent~\(u\in V\) with the following type of improving move.
Let~\(R\subseteq S_u\) and let~\(A\subseteq V\setminus S_u\). Removing the edges between~\(u\) and~\(R\) and adding the edges between~\(u\) and~\(A\) is an improving move if and only if~\(u\) and all nodes in~\(A\) strictly benefit from the whole change\footnote{Note that the allowed strategy changes in the BNE mirror the types of improving moves considered for NE in the NCG. Hence, the BNE is the natural extension of the NE to bilateral edge formation.}.
\item \textbf{Bilateral ($k$-)Strong Equilibrium$^\dagger$(($k$-)BSE):} A strategy vector~\(S\) is in~\emph{Bilateral $k$-Strong Equilibrium ($k$-BSE)}\footnote{For the NCG, Andelman, Feldman, and Mansour~\cite{andelman2009strong} and de Keijzer and Janus~\cite{de2019strong} investigated the Strong Equilibrium (SE), which is stable against unilateral strategy changes by any coalition of agents. The BSE is the natural extension.} if there is no coalition~\(\Gamma\subseteq V\) of size at most~\(k\) such that there is the following type of improving move. The move can delete a subset of edges~\(R\subseteq E\), as long as for each edge~\(uv\in R\) it holds that~\(uv\cap\Gamma\neq\emptyset\). At the same time, it can add a set of new edges~\(A\subseteq 2^V\setminus E\), if for all~\(uv\in A\) it holds that~\(u\) and~\(v\) are both included in~\(\Gamma\). The move is improving if all nodes in~\(\Gamma\) strictly benefit from it.
A graph~\(G\) is in \emph{Bilateral Strong Equilibrium (BSE)} if it is in $n$-BSE.
\end{itemize}
Throughout the different solutions concepts above, we can see the trend of increasingly enhanced coordination, with the BSE admitting the strongest form of agent collaboration.

\textbf{Quality of Equilibria:}
The \emph{social cost} of~\(G\) is defined as the total cost incurred by all agents, denoted by~\(\Cost{G} = \sum_{u\in V}\Cost{u}\). We define the \emph{total distance cost} of $G$ as \(\Dist{G} = \sum_{u\in V}\Dist{u}\) and the \emph{total buying cost} of~\(G\) as \(\buy{G} = \sum_{u\in V}\buy{u}\).
Graphs with minimum social cost for given~\(n\) and~\(\alpha\) are called \emph{social optima}.

For a graph~\(G\) with $n$ nodes and parameter~\(\alpha\), we define the \emph{social cost ratio}~\(\poa{G}\) as follows. Let~\(\opt\) be a social optimum for~\(n\) and~\(\alpha\), then~\(\poa{G}=\frac{\Cost{G}}{\Cost{\opt}}\), i.e., the ratio of the social costs of $G$ and $\opt$. In particular,~\(\poa{G}\) is~\(1\) exactly if~\(G\) is a social optimum.

The \emph{Price of Anarchy (PoA)}~\cite{KP09} for a specific solution concept is a function of~\(n\in\N\) and~\(\alpha\in\R_{>0}\). For given~\(n\) and~\(\alpha\), let~\(\mathcal G\) denote the set of all graphs with~\(n\) nodes which meet the considered equilibrium definition. Then, the PoA for that~\(n\) and~\(\alpha\) is defined as \(\maxset{\poa{G} \mid G\in\mathcal G}\). We will mainly consider the PoA of tree networks, which is defined as the general PoA but there $\mathcal{G}$ is further refined to only contain equilibrium networks that are trees.
The PoA is a measure for the worst-case inefficiency and we are particularly interested in the asymptotics for the different solution concepts, i.e., in bounds on the PoA that depend on the number of agents $n$ and the edge price $\alpha$.

\subsection{Related Work}
The NCG was introduced by Fabrikant, Luthra, Maneva, Papadimitriou, and Shenker~\cite{fabrikant2003network}.
They showed that trees in NE have a PoA of at most~\(5\) and conjectured that there is a constant~\(A\in\N\) such that all graphs in NE are trees if~\(\alpha\geq A\).
While the original tree conjecture was disproven~\cite{albers2014nash}, it was reformulated to hold for~\(\alpha \geq n\). This is best possible, since non-tree networks in NE exist for $\alpha < n$~\cite{MMM13}.
A recent line of research~\cite{albers2014nash, MS10, MMM13, alvarez2017network, AM19, bilo2020tree, dippel2021improved} proved the adapted conjecture to hold for~\(\alpha > 3n-3\), further refined the constant upper bounds on the PoA for tree networks, and showed that the PoA is constant if $\alpha > n(1+\varepsilon)$, for any $\varepsilon > 0$. The currently best general bound on the PoA was established by Demaine, Hajiaghayi, Mahini, and Zadimoghaddam~\cite{demaine2012price}. They derived a PoA bound of~$o(n^\varepsilon)$, for any $\varepsilon>0$,
and gave a constant upper bound for~\(\alpha\in\mathcal{O}(n^{1-\varepsilon})\). Besides bounds on the PoA, the complexity of computing best response strategies and the convergence to equilibria via sequences of improving moves have been studied in ~\cite{fabrikant2003network} and \cite{kawald2013dynamics}, respectively.  

The NCG was also analyzed with regard to Strong Equilibria (SE), for which the PoA was shown to be at most~\(2\) in general~\cite{andelman2009strong} and at most~\(\frac{3}{2}\) for~\(\alpha\geq 2\)~\cite{de2019strong}.
Greedy Equilibria (GE) for the NCG have been introduced by Lenzner~\cite{lenzner2012greedy} who showed that any graph in GE is in 3-approximate NE and that for trees GE and NE coincide. Also, it is shown that if an agent can improve by buying multiple edges, then she can improve by buying a single edge. 

Many variants of the NCG have been studied: a version where agents want to minimize their maximum distance~\cite{demaine2012price}, variants with weighted traffic~\cite{albers2014nash}, or a version on a host graph~\cite{demaine2009price}. Also geometric variants have been considered, for example a variant where the agents want to minize their average stretch~\cite{MSW06} or a version on a host graph with weighted edges~\cite{bilo2019geometric, friedemann2021efficiency}. Moreover, also NCG variants with non-uniform edge prices~\cite{MMO14,CMH14,chauhan2017selfish,bilo2019geometric,BFLMM20,friedemann2021efficiency,BLM22}, locality~\cite{BGLP16,BGLP21,CL15}, and robustness~\cite{MMO15,CLMM16,EFLM20}.

We focus on the BNCG, proposed by Corbo and Parkes~\cite{corbo2005price}. Pairwise Stability dates back to Jackson and Wolinsky~\cite{jackson1996strategic}. So far, the PoA in the BNCG has only been analyzed for PS. In particular, an upper bound of~\(\mathcal{O}(\minset{\sqrt{\alpha}, n/\sqrt{\alpha}})\) was shown~\cite{corbo2005price} and proven to be tight~\cite{demaine2012price}. Corbo and Parkes~\cite{corbo2005price} conjecture that all graphs in NE are also pairwise stable.
Recently, Bilò, Friedrich, Lenzner, Lowski, and Melnichenko~\cite{BFLLM21} study a variant of the BNCG with non-uniform edge price. Also, a version with inverted cost function modeling social distancing was proposed~\cite{friedrich2022social}.

The BNCG has also been generalized by introducing cost-sharing of the edge prices~\cite{albers2014nash}. In this variant, every agent specifies for every possible edge a cost-share she is willing to pay. Then edges with total cost-shares of at least $\alpha$ are formed.
Moreover, Demaine, Hajiaghayi, Mahini, and Zadimoghaddam~\cite{demaine2009price}
introduced the Collaborative Equilibrium (CE), which is in-between PS and SE. A CE is stable against strategy changes by any coalition of agents that concern the joint cost-shares of any single edge.
In contrast to the strategy changes that we focus on, this implies that (coalitions of) agents can also create non-incident edges.   

\subsection{Our Contribution}
We analyze the Bilateral Network Creation Game with regard to different solution concepts in order to evaluate the impact of different levels of cooperation on the Price of Anarchy.

For this, we introduce various natural solution concepts with a wide range of allowed cooperation among the agents. We compare the subset relationships and
we discuss the relationship between solution concepts for the NCG and the BNCG.
Thereby, we disprove the old conjecture by Corbo and Parkes~\cite{corbo2005price} that all graphs in NE are also pairwise stable.

Our main contribution is a thorough investigation of the PoA for various degrees of allowed cooperation among the agents. See \Cref{table:tree_equilibria_results} for an overview. 
\begin{table}[h]
    \centering
    \begin{tabular}{llr}
        \noalign{\hrule height 1pt}
         \textbf{Equilibrium} & \textbf{PoA on Trees} & \textbf{Source} \\
        \noalign{\hrule height 1pt}
         PS &~\(\Theta(\minset{\sqrt{\alpha}, n/\sqrt{\alpha}})\) &~\(O\): \cite{corbo2005price},~\(\Omega\): \cite{demaine2012price}\\
         BSwE &~\(\Theta(\log \alpha)\) & \Cref{sec:swap_on_tree}\\
         BGE &~\(\Theta(\log \alpha)\) & \Cref{sec:greedy_on_tree}\\
         BNE &~\(\Theta(\log \alpha)\), for~\(\alpha\geq n^{1/2+\varepsilon}\), & \Cref{sec:neighborhood_on_tree}\\ &~\(\Theta(1)\), for~\(\alpha\leq \sqrt{n}\) & \Cref{sec:neighborhood_on_tree}\\
         3-BSE &~\(\Theta(1)\) & \Cref{sec:3-SE_on_tree}\\
        \noalign{\hrule height 1pt}
        \textbf{Equilibrium} & \textbf{PoA on General Graphs} & \textbf{Source} \\
        \noalign{\hrule height 1pt}
        BSE &~$\Theta(1)$, for \(\alpha\leq n^{1-\varepsilon}\) & \Cref{chapter:general_equilibria}\\
            &~$\Theta(1)$, for \(\alpha\geq n\log n\) & \Cref{chapter:general_equilibria}\\
             &~$\mathcal{O}\big(\frac{\log n}{\log\log\log n}\big)$, otherwise & \Cref{chapter:general_equilibria}\\
        \noalign{\hrule height 1pt}
    \end{tabular}
    \caption{Asymptotic PoA bounds for~\(\alpha\geq 1\) and~\(\alpha < n^{2-\varepsilon}\), for any $\varepsilon > 0$. \vspace*{-0.5cm}}
    \label{table:tree_equilibria_results}
\end{table}
As a first step in this direction, we mainly focus on tree networks in equilibrium. Such networks are of prime importance in the (B)NCG research, since the existence of equilibria that are trees is guaranteed for most of the parameter space\footnote{For $\alpha \geq 1$ a star on $n$ nodes is an equilibrium for all considered solution concepts.}, and because the PoA of tree networks in the unilateral NCG is constant~\cite{fabrikant2003network}. Moreover, real-world networks are typically sparse and tree-like. Thus, the analysis of tree equilibria serves as a natural first step towards understanding more complicated equilibria.

Most importantly, we show that studying tree networks yields valuable insights on the impact of cooperation for the BNCG. In particular, we observe multiple improvements to the asymptotic PoA as we increase the allowed amount of cooperation. While the PoA is known to be in \(\Theta(\minset{\sqrt{\alpha}, n/\sqrt{\alpha}})\) for PS~\cite{corbo2005price,demaine2012price}, which only allows cooperation for creating a single edge, we show that allowing cooperative edge swaps, as in Bilateral Swap Equilibria or in Bilateral Greedy Equilibria, already improves the PoA to~\(\Theta(\log \alpha)\). 
For~\(\alpha\geq n^{1/2 + \varepsilon}\), with~\(\varepsilon > 0\), we get the same PoA for trees in Bilateral Neighborhood Equilibrium (BNE), which is the bilateral version of the NE in the unilateral NCG. However, on the positive side, we show that the PoA for trees in BNE surprisingly turns constant if edges are cheap, i.e., for~\(\alpha\leq \sqrt{n}\).

Our most significant result is a constant PoA for tree equilibria using the Bilateral 3-Strong Equilibrium. At least on tree networks, this implies that very little agent cooperation is needed to ensure socially good stable states. Thus, for a system designer it is quite easy to enable the agents to escape from socially bad stable states: simply allow cooperation of coalitions of size 3. Contrasting this, we also show that the PoA for Bilateral 2-Strong Equilibria is in $\Omega(\log \alpha)$, i.e., no constant PoA can be achieved by considering smaller coalitions. Thus, we exactly pin-point the minimum amount of cooperation needed to ensure a constant PoA on tree networks.  

For general networks, we investigate the impact of enhanced cooperation on the PoA by studying Bilateral Strong Equilibria. Our results are similar to the best known PoA bounds for NE in the NCG. We prove a constant PoA for Bilateral Strong Equilibria for~\(\alpha\leq n^{1-\varepsilon}\), with~\(\varepsilon > 0\), and for~\(\alpha\geq n\log n\). For the $\alpha$ range in-between, we show that the PoA is in $o(\log n)$.
It remains open whether the PoA is constant for all $\alpha >0$ and we show that our techniques cannot be applied to obtain a constant PoA bound for~\(\alpha = n\).

All omitted details can be found in the appendix.

\section{Relationships of Equilibria}
\label{sec:relationship_between_equilibria}
We explore the subset relationships among the introduced solution concepts. See \Cref{fig:equilibrium_hierarchy} for our results. The relationships follow directly from the definitions, but showing that subsets are proper or showing non-comparability requires a rich set of examples.

\begin{figure}[h]
\centering
\begin{subfigure}{0.45\textwidth}
\centering
\includegraphics[height=5.5cm]{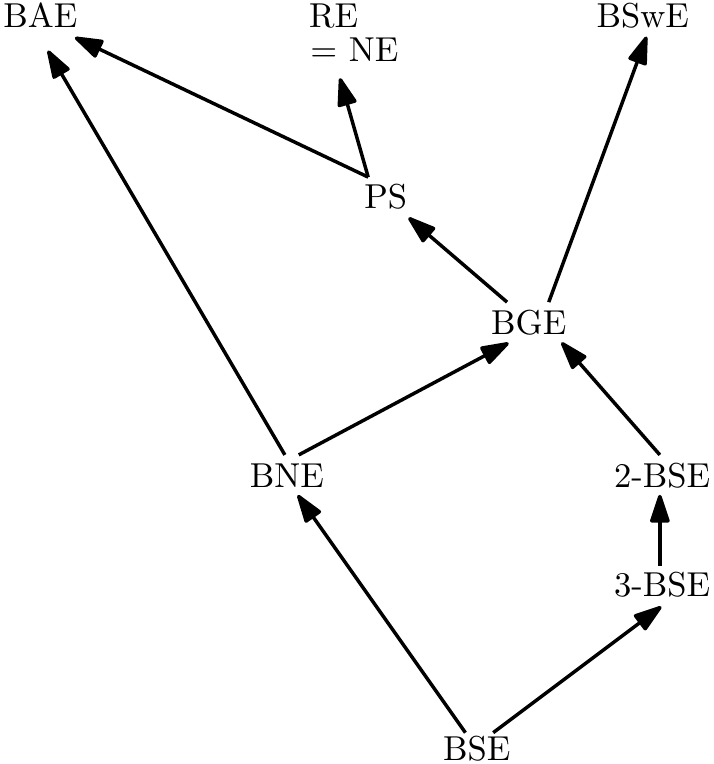}
\caption{The subset relationships between the considered solution concepts. Arrows point from subset to superset, all subset relations are proper.}
\label{fig:equilibrium_hierarchy}
\end{subfigure}
\hfill
\begin{subfigure}{0.45\textwidth}
\centering
\includegraphics[height=5cm]{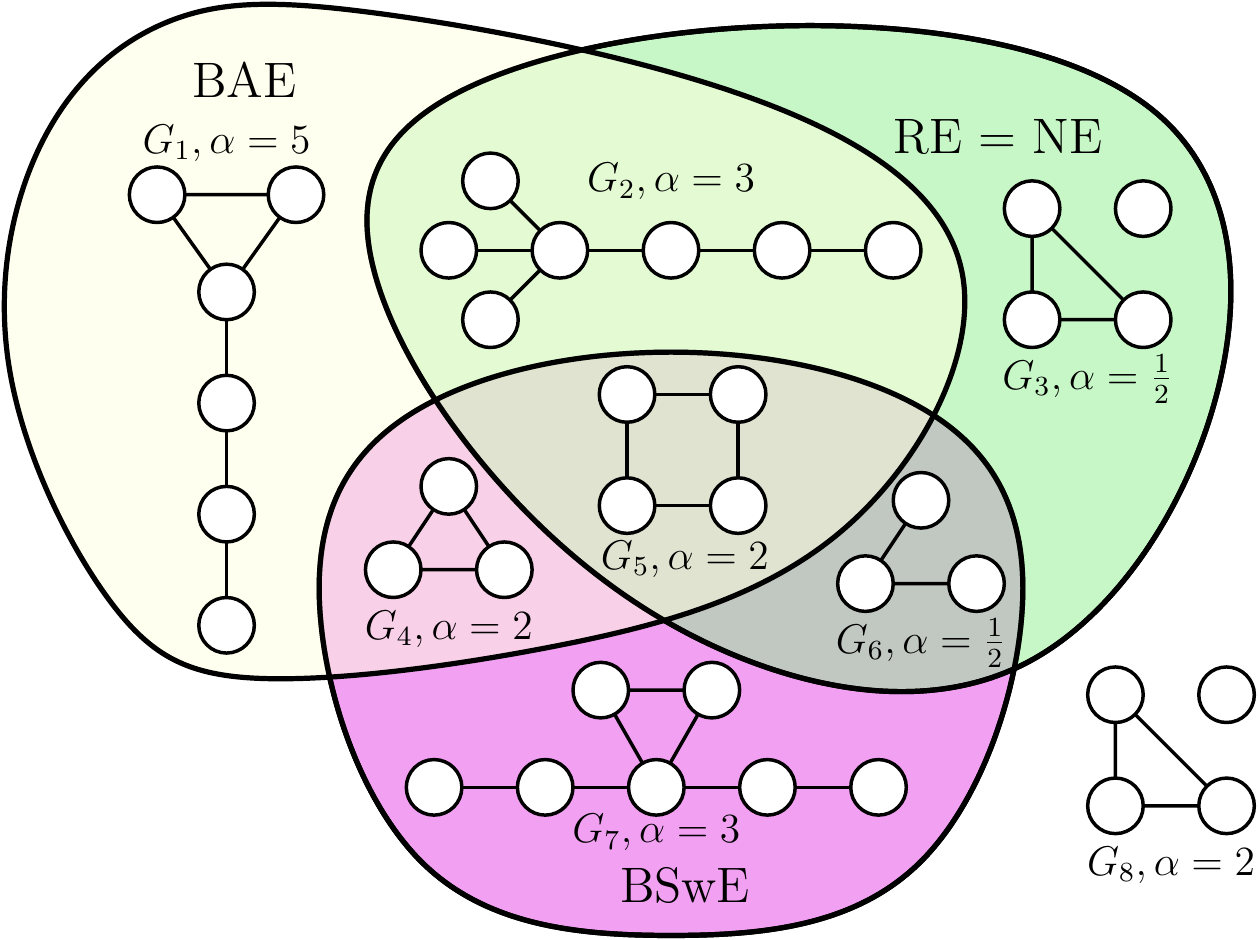}
\caption{Venn diagram showing the relation of RE, BAE, and BSwE. They are pairwise incomparable.}
\label{fig:base_equilibria_venn}
\end{subfigure}
\caption{Overview over the subset relationships between the solution concepts.}
\label{fig:equilibrium_hierarchy_venn}
\end{figure}

We find that the well-known Pairwise Stability~\cite{jackson1996strategic} is a superset of many of the solution concepts that we study, i.e., the BGE and BNE can be understood as stronger refinements of PS.

We compare the NCG and the BNCG and with regard to Remove Equilibria and Add Equilibria.
In the NCG, the corresponding definition of an Add Equilibrium considers that a single agent might add a single edge without any strategy changes of other agents.
For simplicity, we assume for the NCG that each edge of the graph~\(G\) is owned by exactly one incident agent. This allows us to model the \emph{edge assignment} as a function~\(f:E\to V\), where each edge is mapped to one of its incident nodes. Under these assumptions, a graph~\(G\) and edge assignment~\(f\) completely capture the strategy vector of the NCG.

\begin{restatable}{proposition}{thmone}
\label{theorem:uni_bi_add}
Let graph~\(G\) with edge assignment~\(f\) be in Add Equilibrium for the unilateral NCG, then~\(G\) is also in BAE in the BNCG. However, the reverse direction does not hold.
\end{restatable}
As expected, there is no difference regarding Remove Equilibria.
\begin{restatable}{proposition}{thmtwo}
A graph~\(G\) is in Remove Equilibrium in the BNCG exactly if it is in Remove Equilibrium in the unilateral NCG for every edge assignment.
\end{restatable}

This brings us to refuting the conjecture by Corbo and Parkes~\cite{corbo2005price}, which states that every graph~\(G\) in NE in the NCG is also pairwise stable in the BNCG.
For this, we consider a graph in NE which is not in unilateral Remove Equilibrium for a different edge assignment.

\begin{proposition}
\label{theorem:corbo_conjecture_false}
There exists a graph~\(G\) and edge assignment~\(f\) such that the graph with this assignment is in NE in the unilateral NCG but~\(G\) is not pairwise stable in the BNCG.
\end{proposition}
\begin{proof}
\Cref{fig:example_graph_uni_nash_eq_not_bilateral_pair_eq} shows such a graph~\(G\). 
\begin{figure}[h!]
  \centering
  \includegraphics[width=0.25\textwidth]{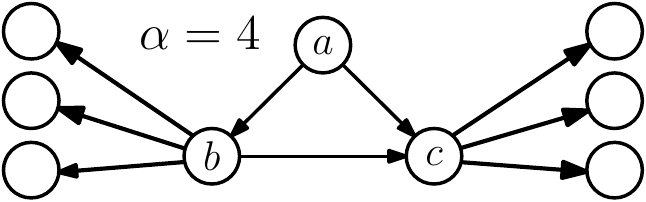}
\caption{Example of a graph that is in NE in the NCG but not PS in the BNCG. The edge owner assignment for the NCG is depicted by edges pointing away from their owner.}
\label{fig:example_graph_uni_nash_eq_not_bilateral_pair_eq}
\end{figure}
While the graph is in NE in the unilateral NCG, it is not pairwise stable in the BNCG because agent~\(b\) profits from removing the edge~\(ba\).
\end{proof}

Another contrast between unilateral and bilateral equilibria is the existence of non-tree equilibrium networks for high $\alpha$.
Let~\(C_n\) denote a cycle of~\(n\) nodes.
\cite{corbo2005price} already showed that for any~\(n \geq 3\), there is a range of~\(\alpha\in\Theta(n^2)\) for which~\(C_n\) is pairwise stable. We show that these~\(\alpha\) ranges can be refined further, so that they even apply for BSE. 
This implies that, contrary to the unilateral NCG, no tree conjecture is possible in the BNCG, as there can be non-tree equilibria for~\(\alpha\in\Theta(n^2)\).

\begin{restatable}{lemma}{lemmafournew}
\label{lemma:cycle_is_strong_eq}
Let~\(n\in\N_{\geq 3}\), then the cycle~\(C_n\) is in BSE for some range of~\(\alpha\in\Theta(n^2)\).
\end{restatable}

\section{Cooperation and the Price of Anarchy}
We present our main results on the impact of the degree of cooperation on the PoA. We start with preliminaries, then analyze the PoA of trees before we investigate general networks. 
\subsection{Preliminaries}
\label{sec:preliminaries}
The social optima of the BNCG have been identified by Corbo and Parkes~\cite{corbo2005price}.
For~\(\alpha< 1\), the clique is the only social optimum and has a cost of~\(\Cost{OPT} = n(n-1)(1+\alpha)\). However, the case with $\alpha < 1$ is of little interest, since it is always beneficial to buy an edge. In the following, we assume~\(\alpha\geq 1\), unless explicitly noted otherwise.

For~\(\alpha\geq 1\), the star is a social optimum, and for~\(\alpha > 1\), it is the only social optimum. There are~\(n-1\) edges, so the total buying cost is~\(2(n-1)\alpha\). The total distance cost among the outer nodes is~\(2(n-1)(n-2)\) while the total distance from and to the center node is~\(2(n-1)\). Hence,
\[\Cost{OPT} = 2(n-1)\alpha + 2(n-1)(n-2) + 2(n-1) = 2(n-1)(\alpha + n-1).\]

Based on a similar proof by Albers, Eilts, Even-Dar, Mansour, and Roditty~\cite{albers2014nash}, we can upper bound the PoA by analyzing the distance cost of a single node. This is very powerful as this implies that it is sufficient to show a small distance cost for a single node, while the buying cost automatically follows suit.

\begin{restatable}{proposition}{theoremfour}
\label{lemma:dist_u_bounds_poa_in_eqRemove}
Let~\(G\) be in RE and connected, then for any node~\(u\in V\) the PoA can be upper bounded by
\(\poa{G} \leq \frac{\alpha + \Dist{u}}{\alpha + n-1}.\qedhere\)
\end{restatable}

Since \(\Dist{u} < n^2\) trivially holds in all connected graphs, we get the following corollary which implies a constant PoA for~\(\alpha\in \Omega(n^2)\).

\begin{restatable}{corollary}{corollaryfive}
\label{lemma:trivial_poa_upper_bound}
Let~\(G\) be a in RE and connected, then~\(\poa{G} \leq 1 + \frac{n^2}{\alpha}\).
\end{restatable}

\subsection{The Price of Anarchy for Tree Networks}
\label{chapter:tree_equilibria}
In the following, we always root the tree~\(G\) at a node~\(r\in V\). Based on this root~\(r\), the \emph{layer} of a node~\(u\in V\) is~\(\layer{u} = \dist{r}{u}\). Each edge~\(uv\in E\) connects two nodes of adjacent layers, i.e., it holds that~\(\card{\layer{u} - \layer{v}} = 1\),
and if~\(\layer{u} = \layer{v} - 1\), we say that~\(u\) is the \emph{parent} of~\(v\) and~\(v\) is the \emph{child} of~\(u\). Each node has exactly one parent, the root~\(r\) has none. The \emph{depth} of a tree is~\(\depth{G} = \max_{v\in V}\layer{v}\), i.e., the maximum distance between the root~\(r\) and any other node.
For~\(u\in V\), let $T_u$ denote the subtree rooted at~\(u\) containing~\(u\) and all of its descendants, i.e., exactly the set of nodes~\(v\in V\) for which the unique~\(v\)-\(r\)-path in the tree contains node~\(u\). In particular, the subtree~\(T_r\) is the original graph~\(G\).
We abuse notation by treating $T_u$ like a node set, so we will write~\(v\in T_u\) when referring to a node in the subtree. We write~\(\depth{T_u}\) when referring to~\(\maxset{\dist{u}{v}\mid v\in T_u}\).

Another important concept for our proofs is the \emph{1-median}~\cite{kariv1979algorithmic} (referred to as center node in \cite{fabrikant2003network}).
A 1-median in a tree is a node~\(u\in V\) with the lowest distance cost. Equivalently, a 1-median can also be defined as a node whose removal from a tree with $n$ nodes will create connected components of size at most $\frac{n}{2}$.
Each tree has exactly one or two 1-medians. When referring to the 1-median of a subtree~\(T\), we also call it a~\emph{\(T\)-1-median}.
In the following proofs, we always consider~\(G\) to be rooted at a 1-median~\(r\in V\). So, for any~\(u\in V\) with~\(u\neq r\), it holds that~\(\card{T_u}\leq\frac{n}{2}\). In particular,
this implies for each non-root~\(u\in V\) that at least~\(\frac{n}{2}\) shortest paths contain~\(r\). Thus, getting closer to~\(r\) can lead to large cost reductions for the nodes. Many of our proofs are based on this.

\subsubsection{Bilateral Swap Equilibrium on Trees}
\label{sec:swap_on_tree}
We upper bound the PoA for trees in BSwE by~\(\Theta(\log \alpha)\). The proof for the asymptotically tight lower bound is presented in~\Cref{sec:greedy_on_tree}. This result shares some parallels with the results for unilateral Asymmetric Swap Equilibria~\cite{mihalak2012asymmetric, ehsani2015bounded}, where it is shown that they have a diameter in~\(\mathcal{O}(\log n)\) and that there is an edge assignment under which a complete binary tree is stable.

A key insight for our upper bound is that subtrees quickly fan out into small subtrees.
\begin{restatable}{lemma}{lemmasix}
\label{lemma:close_median}
If the tree~\(G\) with root and 1-median~\(r\) is in BSwE, then for~\(u\in V\) there is a~\(T_u\)-1-median~\(v\in T_u\) with
\(\layer{v} \leq \layer{u} + \frac{2\alpha}{n}.\qedhere\)
\end{restatable}
\begin{proof}
If there are two~\(T_u\)-1-medians~\(v_1, v_2\in T_u\), then we choose the one closer to \(u\), i.e., we choose~\(v\) with~\(\layer{v} = \minset{\layer{v_1}, \layer{v_2}}\). If~\(u=r\), we have~\(u=v\) and the claim holds.

Otherwise, there is a parent~\(p\in V\) of~\(u\). If~\(u\) is not a 1-median of~\(T_u\), then~\(p\) prefers swapping~\(pu\) for~\(pv\) since that would reduce~\(\dist{p}{T_u}\). 
For agent~\(v\), accepting the proposed swap decreases her distance to~\(r\) by~\(\dist{v}{p} - 1 = \layer{v} - \layer{u}\). By definition of the 1-median, the path from~\(v\) to at least~\(\frac{n}{2}\) nodes must contain~\(r\), hence we get~\((\layer{v} - \layer{u})\frac{n}{2} \leq \alpha\) as~\(G\) is in BSwE. Rearranging this inequality concludes the proof.
\end{proof}

This allows us to obtain a bound on the depth of subtrees.
\begin{restatable}{lemma}{lemmaseven}
\label{lemma:tree_depth_in_swap_eq}
If the tree~\(G\) with root and 1-median~\(r\) is in BSwE, then for~\(u \in V\) we have that 
\(\depth{T_u} \leq \left(1+\frac{2\alpha}{n}\right) \log \card{T_u}.\qedhere\)
\end{restatable}

This already allows us to upper bound the depth and diameter of~\(G\) by~\(\mathcal{O}(\frac{\alpha}{n}\log n)\), which implies a PoA in~\(\mathcal{O}(\log n)\). However, we can improve this bound further for cases where~\(\alpha\) is significantly smaller than~\(n\). The next lemma states that the tree fans out very quickly in the beginning, specifically, all subtrees rooted in layer~\(2\) contain at most~\(\alpha\) nodes.

\begin{restatable}{lemma}{lemmaeight}
\label{lemma:subtree_cardinality_in_add_eq}
If the tree~\(G\) with root and 1-median~\(r\) is in BSwE, then it holds for~\(u\in V\) with~\(\layer{u} \geq 2\) that~\(\card{T_u} \leq \frac{\alpha}{\layer{u} - 1}\).
\end{restatable}

Finally, we can now combine the above lemmas to get a~\(\mathcal{O}(\log\alpha)\) upper bound on the PoA.
This allows us to obtain a~\(\mathcal{O}(\log\alpha)\) upper bound on the PoA. 

\begin{restatable}{theorem}{theoremnine}
\label{theorem:poa_tree_swap_eq}
If the tree~\(G\) is in BSwE, then~\(\poa{G} \leq 2+2\log\alpha\).
\end{restatable}
\begin{proof}
We root~\(G\) at a 1-median~\(r\in V\).
We start by bounding~\(\depth{G}\). If \(\neig{=2}{r} = \emptyset\), we have~\(\depth{G} \leq 1\). Otherwise, we choose~\(u\in \neig{=2}{r}\) such that~\(\depth{T_u}\) is maximal, as this implies that~\(\depth{G} = 2+\depth{T_u}\).
By \Cref{lemma:subtree_cardinality_in_add_eq}, we can upper bound~\(\card{T_u}\) by~\(\alpha\) and with \Cref{lemma:tree_depth_in_swap_eq} it follows that~\(\depth{T_u} \leq \left(1+\frac{2\alpha}{n}\right)\log\alpha\).
Thus, we get
\[\depth{G} \leq 2 + \left(1+\frac{2\alpha}{n}\right)\log\alpha.\]
Now, we upper bound~\(\Dist{r}\) by~\((n-1)\depth{G}\). This gives
\[\Dist{r} \leq (n-1)\left(2 + \left(1+\frac{2\alpha}{n}\right)\log\alpha\right) \leq (n-1)(2 + \log\alpha) + 2\alpha\log\alpha.\]
To conclude the proof, we apply \Cref{lemma:dist_u_bounds_poa_in_eqRemove} to upper bound~\(\poa{G}\) based on~\(\Dist{r}\) with the inequality
\begin{align*}
\poa{G} &\leq \frac{\alpha + \Dist{r}}{\alpha + n-1}
\leq \frac{\alpha + (n-1)(2+\log\alpha) + 2\alpha\log\alpha}{\alpha + n-1}\\
&= \frac{\alpha(1 + 2\log\alpha) + (n-1)(2 + \log\alpha)}{\alpha + n-1}
\leq 2 + 2\log\alpha. \qedhere
\end{align*}
\end{proof}

This result shows that on tree networks swapping an edge is more powerful than only adding or removing an edge. This is good news, as organizing a swap only requires little coordination. However, combining all three operations as considered in the study of (Bilateral) Greedy Equilibria does not grant additional asymptotic improvements, as we shall see next.

\subsubsection{Bilateral Greedy Equilibrium on Trees}
\label{sec:greedy_on_tree}
In \Cref{sec:swap_on_tree}, we have shown that the PoA for BSwE on trees is in~\(\mathcal{O}(\log \alpha)\). This bound also carries over to BGE as they are a subset of BSwE. Now, we show that this bound is asymptotically tight. Note that the lower bound only applies for~\(\alpha\leq n^{2-\varepsilon}\), with~\(\varepsilon>0\), since by~\Cref{lemma:trivial_poa_upper_bound}, for~\(\alpha\in\Omega(n^2)\) the PoA of any connected graph is constant.

To motivate BGE further, we start with showing that on trees they are equivalent to $2$-BSE.
\begin{restatable}{proposition}{theoremten}
Let~\(G\) be a tree. Then~\(G\) is in BGE if and only if it is in $2$-BSE.
\end{restatable}

Now we introduce the \emph{stretched binary tree} that will be used for the PoA bound of~\(\Omega(\log\alpha)\). 
A \emph{stretched binary tree}~\(T\) with parameters~\(d\in\N\) and~\(k\in\N_{\geq 1}\) is defined as follows.
Let~\(B\) be a complete binary tree of depth~\(d\) with root~\(r\). For~\(u\in V_B\setminus\autoset{r}\), we define~\(P_u = \autoset{u^i}_{i\in[k-1]}\cup\autoset{u}\) and define the node set of~\(T\) as~\(V_T = \autoset{r}\cup\bigcup_{u\in V_B\setminus\autoset{r}}P_u\). For~\(uv\in E_B\), where~\(u\) is the parent, the tree~\(T\) contains the edges~\(uv^1, v^1v^2,\dots,v^{k-1}v\). See \Cref{fig:stretched_binary_tree} for an example of a 3-stretched binary tree.
\begin{figure}[h]
\centering
\begin{subfigure}[b]{.45\textwidth}
  \centering
  \includegraphics[width=0.3\textwidth]{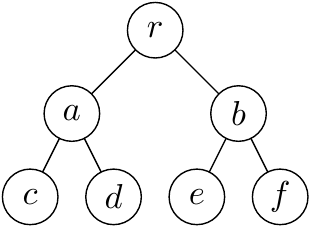}
\caption{Complete binary tree~\(B\) with~\(d=2\).}
\label{fig:complete_binary_tree}
\end{subfigure}\\%
\begin{subfigure}[b]{.45\textwidth}
  \centering
  \includegraphics[width=0.75\textwidth]{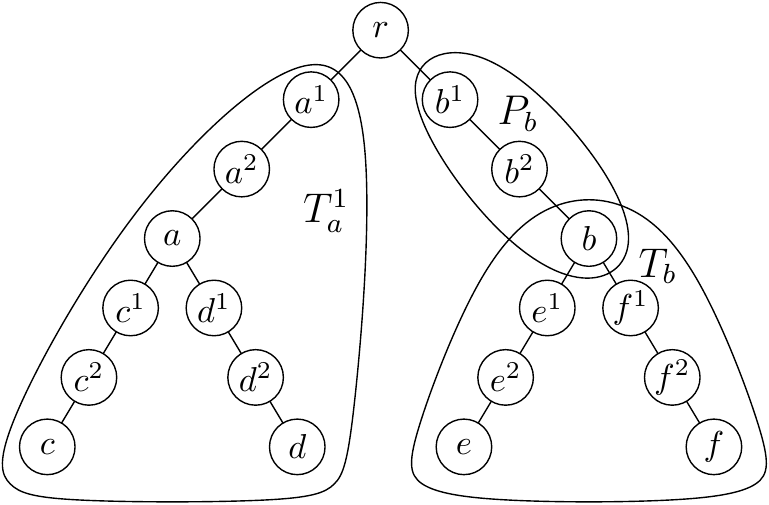}
\caption{Stretched binary tree~\(T\) with~\(d=2, k=3\).}
\label{fig:k_stretched_binary_tree}
\end{subfigure}
\caption{A complete binary tree~\(B\) (top) and the associated 3-stretched binary tree~\(T\) (bottom).}
\label{fig:stretched_binary_tree}
\end{figure}
The resulting graph has \(\left(\card{V_B}-1\right)k+1 = \left(2^{d+1}-2\right)k+1\) nodes and is a tree. For \(u,v\in V_B\), it holds that \(\dist[T]{u}{v} = k\cdot\dist[B]{u}{v}\) and thus \(\depth{T} = k\cdot\depth{B}\).
We also call~\(T\) a~\emph{\(k\)-stretched binary tree}.
For~\(i\in[k-1]\) and~\(u\in B\), we refer to the subtree rooted at~\(u^i\) as~\(T_u^i\), as this is easier to read as~\(T_{u^i}\). Moreover, we use \(u^k\) as an alias for \(u\).

The intuition for using $k$-stretched binary trees is as follows. A binary tree has logarithmic depth in the number of nodes. But since the PoA formula divides the cost by~\(\alpha + n - 1\), a depth of~\(\log n\) will get dominated by the edge cost for~\(\alpha\geq n\log n\). Hence, we stretch the tree to preserve the distance cost. In particular, we will later choose~\(k\in\Theta(\frac{\alpha}{n})\).

\begin{restatable}{proposition}{theoremthirteen}
\label{theorem:stretched_tree_in_Greedy}
Let~\(T\) be a~\(k\)-stretched binary tree, then \(T\) is in BGE for~\(\alpha\geq 7kn\).
\end{restatable}

Finally, we can lower bound the PoA for large~\(\alpha\).
\begin{restatable}{proposition}{theoremfourteen}
For sufficiently large~\(\eta\in\N\) and~\(\eta \leq \alpha\leq \eta^{2-\gamma}\) with~\(0<\gamma \leq 1\), there exists a stretched tree~\(G\) with~\(\frac{\eta}{42} \leq n \leq \frac{\eta}{14}\) nodes such that~\(G\) is in BGE and~\(\poa{G}\geq \frac{25}{32} + \frac{1}{96}\gamma\log\eta\).
\end{restatable}

Since we construct a lower bound for a large range of~\(\alpha\), we need to be able to scale our graph, so we can provide many different values of~\(n\) for a given~\(\alpha\). Thus, the next definition combines many stretched trees into a single graph.

We define a \emph{stretched tree star}~\(G\) with stretch factor~\(k\in\N_{\geq 1}\), target subtree size~\(t\in\R_{\geq 2k+1}\) and target size~\(\eta\in\N_{\geq 2t+1}\) as follows.
Let~\(T\) be a stretched tree with parameter~\(k\) and~\(d\) maximal subject to~\(|T|\leq t\). Then~\(G\) consists of a root~\(r\) with~\(\ceil{\frac{\eta-1}{|T|}}\) copies of~\(T\) as child subtrees.

Finally, we can use stretched tree stars to lower bound the PoA with respect to BGE. 

\begin{restatable}{theorem}{theoremfifteen}\label{theorem:poa_BGE_lowerbound}
For sufficiently large~\(\alpha\) and~\(\eta\in\N_{\geq \alpha}\), a stretched tree star~\(G\) with~\(\eta \leq n \leq \frac{3\eta}{2}\) nodes exists, such that~\(G\) is in BGE and
\(\poa{G}\geq \frac{1}{4}\log\alpha - \frac{17}{8}.\qedhere\)
\end{restatable}
\begin{proof}
	We construct~\(G\) as a stretched tree star with the parameters~\(k=1\),~\(t = \frac{\alpha}{15}\), and~\(\eta\) as provided. The graph~\(G\) is evidently in RE since it is a tree.
	
	In order to show that~\(G\) is in BAE, consider~\(u,v\in V\) with \(\layer{u}\leq\layer{v}\) and~\(uv\notin E\). Moreover, let~\(a,b\in V\) be children of~\(r\) such that \(u,v\in T_a\cup T_b\cup\autoset{r}\). Let~\(T'\) denote the subgraph induced by \(T_a\cup T_b\cup\autoset{r}\). As~\(u\) does not get any closer to the root~\(r\) by adding the edge~\(uv\), it suffices to only consider~\(T'\) in our analysis, which is a binary tree with at most~\(\frac{\alpha}{7}\) nodes (for sufficiently large~\(\alpha\)). Hence, we can apply \Cref{lemma:stretched_binary_tree_in_eqAdd} to conclude that~\(G\) is in BAE, as~\(7k|T'|\leq \alpha\).
	
	Analogously, for a swap~\(u,v,w\in V\) with~\(uv\in E, uw\notin E\) and~\(\layer{u}\leq\layer{w}\), we can again define~\(T'\) and conclude that the change is restricted to this complete binary tree. So, by applying \Cref{lemma:stretched_binary_tree_in_eqSwap}, we conclude that~\(G\) is in BSwE.
	
	Hence, the graph~\(G\) is in BGE.
	It remains to show the logarithmic lower bound on the PoA. \Cref{lemma:stretched_tree_star_poa} provides the lower bound
	\[\poa{G} \geq \frac{nk\left(\log\left(\frac{t}{k}\right) - \frac{9}{2}\right)}{2(\alpha + n - 1)} = \frac{n(\log t - \frac{9}{2})}{2(\alpha + n - 1)}.\]
	We conclude the proof by upper bounding~\(2(\alpha + n - 1)\) in the denominator by~\(4n\) and simplifying as follows:
	\[\poa{G} \geq \frac{n\left(\log t - \frac{9}{2}\right)}{4n} = \frac{\log \frac{\alpha}{15} - \frac{9}{2}}{4} = \frac{1}{4}\log\alpha - \frac{17}{8}.\qedhere\]
\end{proof}
Remember that networks in BGE are a subset of the networks in BSwE and hence the PoA upper bound of $\mathcal{O}(\log \alpha)$ from \Cref{theorem:poa_tree_swap_eq} also holds for networks in BGE. Thus, \Cref{theorem:poa_BGE_lowerbound} establishes a tight bound on the PoA for tree networks in BGE.  

\subsubsection{Bilateral Neighborhood Equilibrium on Trees}
\label{sec:neighborhood_on_tree}
Here we prove that the PoA of BNE is in~\(\Theta(\log \alpha)\) for~\(\alpha\geq n^{1/2+\varepsilon}\), with~\(\varepsilon > 0\), so it remains asymptotically unchanged in comparison to BSwE and BGE. However, the PoA surprisingly changes to being constant for~\(\alpha\leq\sqrt{n}\). Thus, the additional coordination improves the asymptotic PoA for small values of~\(\alpha\). For omitted details see \Cref{app:neighborhood_on_tree}.

Since networks in BNE are also in BSwE, the PoA upper bound of $\mathcal{O}(\log \alpha)$ from \Cref{theorem:poa_tree_swap_eq} carries over. Now we show that this bound is tight for~\(\alpha\geq n^{1/2+\varepsilon}\), with~\(\varepsilon > 0\). For this, we again use stretched tree stars, but this time we have to check for stability with respect to BNE. 
\begin{restatable}{lemma}{lemmasixteen}
\label{lemma:stretched_tree_star_neighborhood}
Let~\(G\) be a stretched tree star with parameter~\(k\in\N\) based on a stretched tree~\(T\). Let~\(k=1\) or~\(\alpha\geq 6kn\). Then the graph~\(G\) is in BNE if
\(\frac{3n\cdot\depth{G}}{\alpha}+1 \leq \frac{\alpha}{3\card{T}\cdot\depth{G}}.\qedhere\)
\end{restatable}

Now, we can use \Cref{lemma:stretched_tree_star_neighborhood} to show a lower bound on the PoA for BNE on trees.

\begin{restatable}{theorem}{theoremseventeen}
\label{theorem:neighborhood_tree_large_alpha}
The following statements hold:
\begin{itemize}
 \item[(i)] For~\(0<\varepsilon \leq 1\), sufficiently large~\(\eta\in\N\) and~\(9\eta \leq \alpha\leq \eta^{2-\varepsilon}\), there exists a BNE~\(G\), with~\(\eta \leq n \leq \frac{3}{2}\eta\), such that the inequality
\(\poa{G}\geq \frac{\varepsilon}{168}\log \alpha - \frac{3}{28}\) holds.
\item[(ii)] For~\(0<\varepsilon \leq \frac{1}{2}\), sufficiently large~\(\eta\in\N\) and~\(\eta^{1/2+\varepsilon} \leq \alpha \leq \eta\), there exists a BNE~\(G\), with~\(\eta \leq n \leq \frac{3}{2}\eta\), such that the inequality
\(\poa{G}\geq \frac{1}{4}\varepsilon\log\alpha - \frac{9}{8}\) holds.
\end{itemize}
\end{restatable}

Looking at the range for~\(\alpha\) in \Cref{theorem:neighborhood_tree_large_alpha},
we see that no lower bound for~\(\alpha\leq \sqrt{n}\) is derived. In fact, we cannot apply \Cref{lemma:stretched_tree_star_neighborhood} for~\(\alpha\leq\sqrt{n}\), as inserting~\(\alpha = \sqrt{n}\) gives the inequality
\[3\sqrt{n}\cdot\depth{G} + 1 \leq \frac{\sqrt{n}}{3\card{T}\cdot\depth{G}},\]
which does not hold for any legal parameters.
This is not a fault of our technique, but instead, we show that the PoA is actually constant for this range of~\(\alpha\).

\begin{restatable}{theorem}{theoremnineteen}\label{theorem:BNE_constant_poa}
Let~\(G\) be a tree with $n>15$. If~\(G\) is in BNE for~\(\alpha \leq \sqrt{n}\), then~\(\poa{G} \leq 4\).
\end{restatable}
\begin{proof}
Let~\(r\) denote the 1-median and root of~\(G\).
Let~\(u\in V\) be a node of maximum layer. If~\(\layer{u} \leq 2\), the claim holds. So, we assume that~\(\layer{u} \geq 3\). Moreover, it holds that~\(\alpha < \frac{n}{2}\).
We assume that~\(n > 4\) and thus~\(\alpha < \frac{n}{2}\), as otherwise the diameter is at most~\(3\) and the claim holds.

With~\(i = \card{\neig{=3}{r}}\), let~\(\autoset{c_j}_{j\in[i]}\) denote the nodes in layer~\(3\) sorted descendingly by their subtree size. Let~\(k=\minset{\floor{\frac{n}{\alpha}-1}, i}\).
We consider the following change around~\(u\). Agent~\(u\) buys an edge towards~\(r\) and towards the nodes in~\(\autoset{c_j}_{j\in[k]}\setminus\autoset{u}\).

Agent~\(u\) pays for at most~\(\frac{n}{\alpha}\) additional edges, so her buying cost increases by at most~\(n\). Connecting to~\(r\) decreases her distance cost by at least~\(2\frac{n}{2}\). Further, agent~\(u\) profits from the new direct connections to the nodes in~\(\autoset{c_j}_{j\in[k]}\setminus\autoset{u}\), so the overall change is beneficial for~\(u\). Each agent~\(c\in\autoset{c_j}_{j\in[k]}\) profits because her distance to~\(r\) decreases by~\(1\), so her distance cost decreases by at least~\(\frac{n}{2}\) while she only has to pay for one additional edge.
Then, agent~\(r\) must not benefit from the proposed change, as~\(G\) is in BNE.
As agent~\(r\) decreases her distance to each node in~\(c\in\autoset{c_j}_{j\in[k]}\setminus\autoset{u}\) by~\(1\) and to~\(u\) by~\(2\), it must hold that
\(\sum_{j=1}^k \card{T_{c_j}} < \alpha.\)
Then, if \(i > k\), this allows us to upper bound \(\card{T_{c_{k+1}}}\) as follows
\[\card{T_{c_{k+1}}} < \frac{\alpha}{k} < \frac{\alpha}{\frac{n}{\alpha}-2} \leq \frac{\alpha}{\frac{n}{2\alpha}} = 2.\]
So, if the node \(c_{k+1}\) exists, then it has no children.
Thus, only the agents \(c\in\autoset{c_j}_{j\in[k]}\) have descendants beyond layer~\(3\). Carrying over from BGE, we get for \(c\in\autoset{c_j}_{j\in[k]}\) that \(\depth{T_c} \leq \log\alpha\).
We conclude that \(\Dist{r} \leq 3(n-1) + \alpha\log\alpha\) and apply \Cref{lemma:dist_u_bounds_poa_in_eqRemove} to get
\[\poa{G} \leq \frac{\alpha + 3(n-1) + \alpha\log \alpha}{\alpha + (n-1)} \leq 3 + \frac{\sqrt{n}\left(1 + \log\sqrt{n}\right)}{n} \leq 4.\qedhere\]
\end{proof}
Note that \Cref{theorem:BNE_constant_poa} at least partially recovers a well known positive result from the unilateral NCG using the NE as solution concept: that the PoA of tree networks is constant. While for the unilateral NCG with NE this holds for all edge prices $\alpha$, we get the contrasting result that for the BNCG using the BNE, this is only true for $\alpha\leq \sqrt{n}$. However, in the following we will see that we actually can guarantee a constant PoA on trees for all $\alpha$ if we allow coalitions of size 3 to cooperate.

\subsubsection{Bilateral 3-Strong Equilibrium on Trees}
\label{sec:3-SE_on_tree}
We show that the PoA for 3-BSE on trees is constant. Thus, allowing coalitions of three agents provides us with the same asymptotics as NE in the unilateral NCG, at least on trees.

The intuition behind this result is derived from \Cref{sec:neighborhood_on_tree}, in particular from the PoA lower bound for BNE. In the proof of~\Cref{lemma:stretched_tree_star_neighborhood}, we encountered a situation where an agent~\(u\in V\) from a deep layer attempted to connect to a node~\(v\in V\) from a layer closer to the root. Adding the edge~\(uv\) would reduce the distance cost of agent~\(u\) significantly, so agent~\(u\) was willing to buy many extra edges to incentivize agent~\(v\) to accept the connection, but it ultimately failed to provide enough value to offset the increased buying cost of agent~\(v\).

For 3-BSE we consider the following: What if agent~\(u\) does not need to convince agent~\(v\) to buy an extra edge, but instead to swap an existing one? The swap is possible since agents~\(u\) and~\(v\) are part of a coalition, and agent~\(u\) can collaborate with a third member of the coalition to provide the incentive. This idea leads us to our key lemma, which states that all but one child-subtrees of a node must be shallow.

\begin{restatable}{lemma}{lemmatwenty}
\label{lemma:3SE_shallow_subtrees}
Let~\(G\) be a tree in 3-BSE with root and~\(1\)-median~\(r\), then every node~\(u\in V\) has at most one child~\(c\in V\) such that \(\depth{T_c} > 2\ceil{\frac{4\alpha}{n}} + 1\).
\end{restatable}
\begin{proof}
Assume towards a contradiction that there are two children~\(c, c'\) of node~\(u\) whose subtrees have at least depth~\(2\ceil{\frac{4\alpha}{n}} + 2\).
Then, there is a node~\(z\in T_c\) such that~\(\layer{z} = \layer{u} + 2\ceil{\frac{4\alpha}{n}} + 3\). Moreover, on the path from~\(c\) to~\(z\), there are nodes~\(x,y\in T_c\) such that~\(y\) is the child of~\(x\) and~\(\layer{x} = \layer{u} + \ceil{\frac{4\alpha}{n}} + 2\). Analogously, we define the nodes~\(x', y', z'\in T_{c'}\). These paths are visualized in \Cref{fig:3SE_two_long_paths} along with an annotation of their layers relative to~\(u\).
\begin{figure}[h]
\centering
\includegraphics[width=0.3\textwidth,page=1]{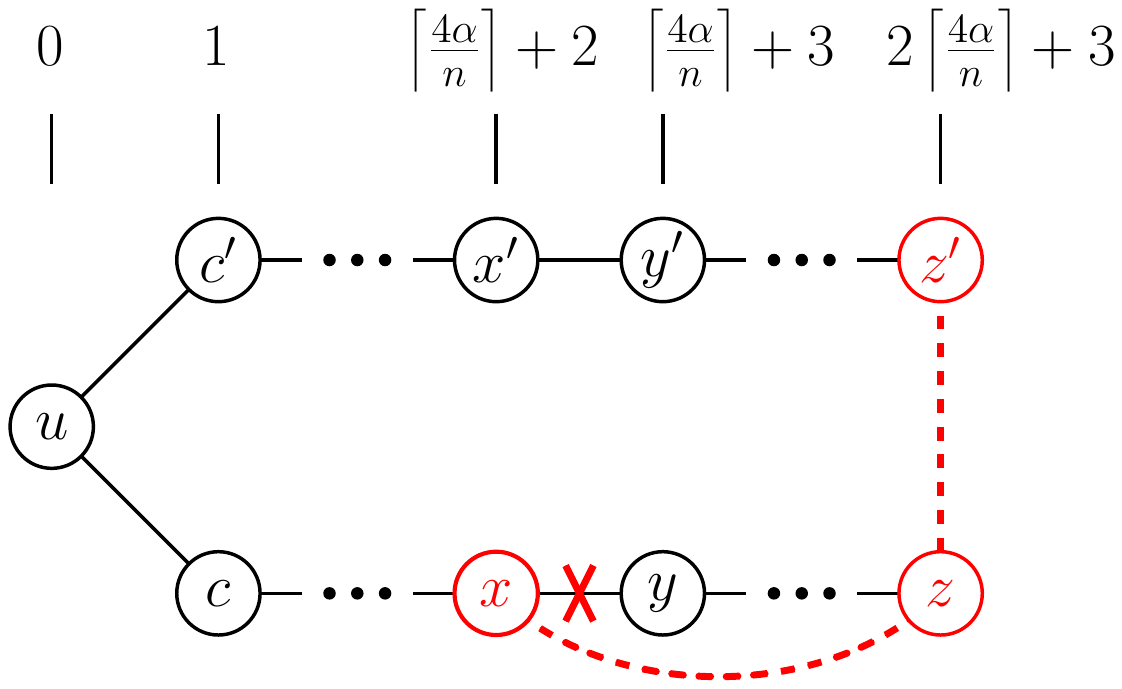}
\caption{The paths considered in the proof of \Cref{lemma:3SE_shallow_subtrees}. The proposed coalitional change is indicated in red (dashed red lines are to be built, the edge $xy$ is to be removed). Nodes outside of the two paths are omitted for simplicity. The labels on top show the layer of the nodes relative to the layer of~\(u\).}
\label{fig:3SE_two_long_paths}
\end{figure}

We now consider a coalitional move by agents~\(x\),~\(z\) and~\(z'\), in which the edges~\(xz\) and~\(zz'\) get added while edge~\(xy\) is removed.
This move decreases the distance from~\(z\) to~\(x\) by~\(\ceil{\frac{4\alpha}{n}}\) and by extension the distance to~\(u\). Since the path from~\(z\) to at least \(\frac{n}{2}\) nodes must contain~\(u\) and since the path to~\(x\) also gets shorter, the distance cost for~\(z\) decreases by at least \[\ceil{\frac{4\alpha}{n}}\cdot \left(\frac{n}{2}+1\right) > 2\alpha.\] Thus, even when considering the increased buying cost, agent~\(z\) profits from this coalitional move.
Likewise, the distance cost for~\(z'\) decreases by at least \[\left(\ceil{\frac{4\alpha}{n}}-1\right)\cdot \frac{n}{2} > \alpha,\] so~\(z'\) also profits by the proposed change.

As~\(G\) is in 3-BSE, agent~\(x\) must not benefit from the coalition. The buying cost of agent~\(x\) remains unchanged, so her distance cost must not improve. The only nodes towards which the distance from~\(x\) may increase are the nodes in~\(T_y\). This increase per node can be upper bounded by the distance increase towards~\(y\), which is \(\ceil{\frac{4\alpha}{n}}\) as the new shortest path goes through~\(z\).
On the other hand, the distance from~\(x\) to~\(x'\) decreases from \(2\ceil{\frac{4\alpha}{n}}+4\) down to~\(2+\dist{z'}{x'} = \ceil{\frac{4\alpha}{n}}+3\), this is an improvement by \(\ceil{\frac{4\alpha}{n}}+1\). And by extension, the distance to each node in~\(T_{x'}\) also improves by at least this amount.

As we assume that the distance cost of agent~\(x\) is not reduced by the coalitional move, we deduce that \[\left(\ceil{\frac{4\alpha}{n}}+1\right) \card{T_{x'}} \leq \ceil{\frac{4\alpha}{n}} \card{T_y},\] which implies that \(\card{T_{x'}} < \card{T_y}\). But, by symmetry, we can also consider the coalition~\(x', z', z\), which gives us the bound \(\card{T_x} < \card{T_{y'}}\). In combination with the inequalities \(\card{T_y} < \card{T_x}\) and \(\card{T_{y'}} < \card{T_{x'}}\), we conclude that this is a contradiction, so there cannot be two different child-subtrees with a depth greater than \(2\ceil{\frac{4\alpha}{n}} + 1\).
\end{proof}

With this insight, we now proceed to show a constant PoA.
\begin{theorem}
\label{lemma:3SE_root_dist_cost}
Let~\(G\) be a tree in 3-BSE, then~\(\poa{G}\leq 25\).
\end{theorem}
\begin{proof}
Let \(d=\depth{G}\) and \(r\in V\) be the 1-median and root of \(G\). For a node~\(u\) with~\(\layer{u}=d\), consider the path from~\(r\) to~\(u\) denoted by \(r=u^0,u^1,\dots,u^d=u\). For~\(i\in [d]\), we denote the subtree rooted in~\(u^i\) by~\(T_u^i\) for better readability.

We derive an upper bound on~\(\Dist{r}\) based on the path to~\(u\). To do so, for any~\(i,j\in[d]\) with~\(i+j\leq d\), we split \(\dist{r}{T_u^i}\) into its two components \(\dist{r}{T_u^{i+j}}\) and \(\dist{r}{T_u^i\setminus T_u^{i+j}}\).

For any node~\(v\in T_u^i\setminus T_u^{i+j}\), there is a~\(k\in[i,i+j]\) such that~\(u^k\) is the deepest common ancestor of~\(u\) and~\(v\). By \Cref{lemma:3SE_shallow_subtrees}, we have that the layer of~\(v\) is at most \(k + 2\ceil{\frac{4\alpha}{n}} + 2\), as otherwise~\(u^k\) would have two child-subtrees with a layer exceeding the threshold. 
To keep our formulas simple, we define \(x = 2\ceil{\frac{4\alpha}{n}} + 2\) and follow that~\(\layer{v} \leq i+j+x\).

Inserting this into our upper bound yields the following recursive inequality
\[\dist{r}{T_u^i} \leq \dist{r}{T_u^{i+j}} + \card{T_u^i\setminus T_u^{i+j}}(i+j+x).\]

We extend our notation for~\(i>d\), such that~\(T_u^i\) denotes an empty subgraph. This allows us to transform the recursion into an infinite sum. Moreover, we fix~\(i = 0\) since we are interested in \(\Dist{r} = \dist{r}{T_u^0}\). Thus, we get
\begin{align*}
	\Dist{r}&\leq \sum_{k=0}^\infty \card{T_u^{kj}\setminus T_u^{(k+1)j}}((k+1)j+x)\\
	&\leq xn + j\sum_{k=0}^\infty \card{T_u^{kj}\setminus T_u^{(k+1)j}}(k+1).
\end{align*}

Now, we choose \(j = \ceil{\frac{2\alpha}{n}} + 1\), as this ensures, by \Cref{lemma:close_median}, for any~\(i\in\N\) that ~\(T_u^{i+j}\) contains at most half as many nodes as~\(T_u^i\).
While this already gives us the bound \(\big|T_u^{kj}\big| \leq 2^{-k}n\), we insert \(2^{-(k+1)}n\) for \(\big|T_u^{kj}\setminus T_u^{(k+1)j}\big|\) in the sum since this represents the worst case where the deep trees contain as many nodes as possible. This yields
\[\Dist{r}\leq xn + j\sum_{k=0}^\infty 2^{-(k+1)}n(k+1) = xn + jn\sum_{k=1}^\infty 2^{-k}k = xn +2jn.\]
Substituting~\(x\) and~\(j\) by their underlying values finally yields the result
\begin{align*}
	\Dist{r}&\leq \left(2\ceil{\frac{4\alpha}{n}} + 2\right)n + 2\left(\ceil{\frac{2\alpha}{n}} + 1\right)n\\
	&\leq 16\alpha + 2n + 8\alpha + 2n = 24\alpha + 4n.
\end{align*}
For~\(n<5\), the proposed bound holds for any tree, and for~\(n\geq 5\), we can upper bound~\(4n\) by~\(5(n-1)\). Finally, applying \Cref{lemma:dist_u_bounds_poa_in_eqRemove} yields
\[\poa{G} \leq \frac{\alpha + 24\alpha + 5(n-1)}{\alpha + n - 1} \leq 25.\qedhere\]
\end{proof}
Thus, we have proven that already a very limited form of cooperation, i.e., joint coalitional moves of coalitions of size at most 3, guarantees a constant PoA for tree networks. This is a strongly positive result since from an agents' point of view, negotiating changes within such small coalitions seems feasible. Such coalitional moves might be much easier to coordinate than an improving move of some agent in the BNE setting, where potentially many agents must simultaneously evaluate and agree to the change.  

\subsection{The Price of Anarchy for General Networks}
\label{chapter:general_equilibria}
We now consider general graphs. This means our arguments cannot rely anymore on reducing the distance to a~\(1\)-median or on the path between two nodes being unique. This makes reasoning about the equilibria significantly more difficult.

As a first step in this direction, we investigate if cooperation of the agents can guarantee good equilibria at all, i.e., we focus our analysis on the most powerful solution concept, the BSE, for which strategy changes by coalitions of arbitrary size are permitted, as long as all members of the coalition benefit from the change. We show that the PoA is constant for~\(\alpha\geq n\log n\) and for~\(\alpha\leq n^{1-\varepsilon}\), with~\(\varepsilon > 0\), but demonstrate that our technique cannot yield constant bounds for~\(\alpha=n\).

We start by showing that the set of social optima and BSE coincide for~\(\alpha\leq 1\). Then, for~\(\alpha>1\),  we show a constant PoA for most ranges of~\(\alpha\).

It is already known that for~\(\alpha\leq 1\) any pairwise stable graph is socially optimal~\cite{corbo2005price}, so any BSE also must be socially optimal. Hence, the following proposition shows that BSE exist for such~\(\alpha\). Moreover, it shows that the situation for~\(\alpha > 1\) is more complicated.

\begin{restatable}{proposition}{theoremtwentytwo}
For~\(\alpha< 1\), the clique is the only BSE. For~\(\alpha=1\), only graphs with diameter at most~\(2\) are in BSE. For~\(\alpha > 1\), the star is in BSE, but so are other graphs.
\end{restatable}

Next, we consider the PoA for~\(\alpha>1\). For any given~\(n\) and~\(\alpha\), we show that the agent with the highest cost in an arbitrary graph implies an upper bound on the PoA. The considered graph does not even need to be in BSE. Consequently, it suffices to identify graphs where all agents have a low cost in order to bound the PoA.

\begin{restatable}{lemma}{lemmatwentythree}
\label{lemma:worst_node_bounds_se_cost}
Let~\(G\) be a graph with~$n>1$ and let~\(u\in V_G\) be the agent with the highest cost. Then, for every graph~\(H\) in BSE with~\(n\) agents and the same~\(\alpha\) it holds that
\(\poa{H} \leq \frac{\Cost[G]{u}}{\alpha + n - 1}.\qedhere\)
\end{restatable}

Equipped with \Cref{lemma:worst_node_bounds_se_cost}, we can now construct graphs where the worst-off agent has a low cost. We do so by building trees in which the cost is distributed evenly over the agents.

\begin{restatable}{lemma}{lemmatwentyfour}
\label{lemma:worst_node_in_regular_tree}
For~\(d\in \N_{\geq 2}\) and~\(n\in\N\), let~\(G\) be an almost complete~\(d\)-ary tree with~\(n\) agents, then for all~\(u\in V\) holds that \(\Cost{u}\leq  (d+1)\alpha + 2(n-1)\log_d n\).
\end{restatable}

Now we can set~\(d\) depending on~\(\alpha\). For \(\alpha\geq n\log n\), it suffices to consider binary trees to obtain a good PoA, as the logarithmic distances get dominated by the buying cost.

\begin{restatable}{theorem}{theorembeforetwentysix}
\label{theorem:se_poa_large_alpha}
Let~\(G\) be in BSE and~\(\alpha \geq n\log n\), then~\(\poa{G} \leq 5\).
\end{restatable}
\begin{proof}
We apply \Cref{lemma:worst_node_in_regular_tree} with~\(d = 2\) and combine it with \Cref{lemma:worst_node_bounds_se_cost} to get 
\[\poa{G} \leq \frac{3\alpha + 2(n-1)\log n}{\alpha + n - 1} \leq 3 + \frac{2(n-1)\log n}{n\log n} \leq 3 + 2 = 5.\qedhere\]
\end{proof}

For~\(\alpha\) significantly smaller than~\(n\), it is necessary to keep the distances small. This can only be achieved by scaling the node degrees. This is possible since the buying cost is low.

\begin{restatable}{theorem}{theoremtwentysix}
For~\(\varepsilon\in\R_{>0}\) and~\(\alpha \leq n^{1-\varepsilon}\), let~\(G\) be in BSE, then~\(\poa{G} \leq 3+\frac{2}{\varepsilon}\).
\end{restatable}
\begin{proof}
We apply \Cref{lemma:worst_node_in_regular_tree} with~\(d = \ceil{ n^\varepsilon}\) and combine it with \Cref{lemma:worst_node_bounds_se_cost} to get
\[\poa{G} \leq \frac{\left(\ceil{n^\varepsilon}+1\right)\alpha + 2(n-1)\log_{\ceil{ n^\varepsilon}} n}{\alpha + n - 1}.\]
We upper bound the first summand by using~\(\ceil{ n^\varepsilon} + 1\leq n^\varepsilon + 2\) and inserting the bound for~\(\alpha\) to get that~\(\left(n^\varepsilon + 2\right)\alpha \leq n + 2n^{1-\varepsilon}\).
For the second summand, we do a change of base to convert~\(\log_{n^\varepsilon} n = \frac{\log n}{\log{n^\varepsilon}} = \frac{1}{\varepsilon}\) and upper bound with~\(\frac{2(n-1)}{\varepsilon}\). This gives us the inequality
\[\poa{G} \leq \frac{n + 2n^{1-\varepsilon} + \frac{2(n-1)}{\varepsilon}}{\alpha + n - 1}.\]
Since~\(\alpha \geq 1\), we lower bound the denominator by~\(n\) and conclude the proof with
\[\poa{G} \leq 1 + 2n^{-\varepsilon} + \frac{2}{\varepsilon} \leq 3 + \frac{2}{\varepsilon}.\qedhere\]
\end{proof}

This bound, however, is only constant with regard to a specific~\(\varepsilon\), so~\(\frac{2}{\varepsilon}\) can become arbitrarily large. Moreover, the preceding two theorems still leave a gap for~\(n^{1-\varepsilon} < \alpha < n\log n\), which is similar to the gap for NE in the NCG, for which no constant PoA has yet been established~\cite{AM19}.
By using binary trees, we can show a PoA in \(o(\log n)\) for general~\(\alpha\).

\begin{theorem}
\label{theorem:strong_eq_general_bound}
Let~\(G\) be in BSE, then \[\poa{G} \leq 2 + \log\log n + 2\frac{\log n}{\log\log\log n}.\qedhere\]
\end{theorem}
\begin{proof}
We apply \Cref{lemma:worst_node_in_regular_tree} with~\(d = \ceil{\log\log n}\) and combine it with \Cref{lemma:worst_node_bounds_se_cost} to get
\begin{align*}
	\poa{G}&\leq \frac{(\ceil{\log\log n}+1)\alpha + 2(n-1)\log_{\ceil{ \log\log n}} n}{\alpha + n - 1}\\
	&\leq \log\log n + 2 + \frac{2(n-1)\log_{\log\log n} n}{n - 1}.
\end{align*}
By a change of base, we convert~\(\log_{\log\log n} n\) to~\(\frac{\log n}{\log\log\log n}\). The proof concludes with
\begin{align*}
	\poa{G} &\leq 2 + \log\log n + \frac{2(n-1)\frac{\log n}{\log\log\log n}}{n - 1}\\
	&= 2 + \log\log n + 2\frac{\log n}{\log\log\log n}. \qedhere
\end{align*}
\end{proof}
\Cref{theorem:strong_eq_general_bound} gives a~\(\mathcal{O}(\frac{\log n}{\log\log\log n})\) bound which is the same as~\(\mathcal{O}(\frac{\log \alpha}{\log\log\log \alpha})\) for the remaining gap~\(n^{1-\varepsilon} < \alpha < n\log n\). We suspect that additional improvements can be achieved by refining the parameter~\(d\).
Nonetheless, it remains open whether there is a constant PoA for~\(\alpha\) close to~\(n\). What we do know, however, is that such a bound cannot be obtained by our technique. The reason is that for a constant PoA for~\(\alpha = n\), the cost cannot remain evenly distributed across the agents as we scale-up~\(n\), as we see now.

\begin{restatable}{proposition}{theoremtwentynine}
There is no constant~\(p\in\R\) such that for all~\(n\in\N\) and~\(\alpha=n\) there exists a graph~\(G\) such that for all agents~\(u\in V\) it holds that \(\frac{\Cost[G]{u}}{\alpha + n - 1} \leq p\). 
\end{restatable}
As an extension of this proposition, a constant PoA for~\(\alpha = n\) would imply that some graphs can only reach an equilibrium state through a series of multiple improving coalitional moves.
In particular, we have already seen that for~\(\alpha = n\), every agent in an almost complete binary tree has costs in~\(\Theta(n\log n)\). But if the PoA is constant, any graph in BSE must have a node with degree~\(\Omega(\sqrt[4p]{n})\) and thus costs in~\(\Omega(n\sqrt[4p]{n})\), which is asymptotically worse. Evidently, this transition cannot be achieved in a single move, as only agents within the coalition can increase their degree, but they would only do so if they benefit from the change.

\section{Conclusion and Outlook}
We analyzed the BNCG under different amounts of agent cooperation. Previously, only Pairwise Stability, where cooperation is restricted to single edge additions, has been studied.

On tree networks our results convey the general picture that the PoA improves asymptotically as we progress towards more cooperation among the agents. In particular, the significant improvements achieved by BSwE, BGE, BNE and 3-BSE give valuable insights for system designers.
When defining the protocol or contracts by which agents establish and remove connections, a system designer should try to allow for edge swaps or even joint changes by coalitions of size at least~\(3\), instead of only permitting single edges to be removed or added.
This is a positive result, as coordinating a swap or a~\(3\)-ways contract is presumably significantly less demanding than coordinating joint changes by larger coalitions.

Our results for tree networks raise the most pressing open question of whether these bounds on the PoA also carry over to general networks. For the latter, we showed that BSE have a constant PoA for most ranges of~\(\alpha\) and we conjecture that this extends to the full range of~\(\alpha\).  However, coordinating simultaneous infrastructural changes by large coalitions might prove an inconceivable effort in many practical applications. Hence, settling the PoA for solution concepts with more restricted agent coordination, like the 3-BSE, seems even more relevant.

\bibliographystyle{ACM-Reference-Format}
\bibliography{references}


\appendix

\section{Omitted Details from Section~\ref{sec:relationship_between_equilibria} }
\label{appendix:equilibria_relationships}
\textbf{Remove Equilibria, Bilateral Add Equilibria, and Bilateral Swap Equilibria: }
We show that RE, BAE, and BSwE are truly distinct. The examples given in~\Cref{fig:base_equilibria_venn} imply the following proposition.
\begin{proposition}
For each combination of RE, BAE, and BSwE, there is a graph that is stable for exactly that combination.
\end{proposition}
Interestingly, Remove Equilibria and Pure Nash Equilibria are identical. This essentially follows from an observation by Corbo and Parkes~\cite{corbo2005price}, that removing a single edge is just as powerful as removing multiple edges at once. This carries over from the unilateral NCG.
\begin{proposition}
\label{theorem:remove_eq_is_nash_eq}
The set of Remove Equilibria coincides with the set of Pure Nash Equilibria. 
\end{proposition}
\begin{proof}
NE are trivially a subset of RE, as removing an edge is performed by changing the strategy of a single agent. Thus, it remains to show that if a graph~\(G\) is in RE, then it is also in NE.

Assume for the contrapositive that~\(G\) is not in NE. Then there is some agent~\(u\in V\) and a best-response strategy~\(S_u'\subseteq V\) such that agent~\(u\) can improve by changing to strategy~\(S_u'\).
We do a case distinction based on whether~\(S_u'\) is a subset of~\(S_u\).

If~\(S_u'\) adds a new target node~\(v\in S_u'\setminus S_u\), then the new edge~\(uv\) is only built if~\(u\) is already included in~\(S_v\). If this is the case, then~\(G\) is not in RE, as agent~\(v\) can improve by removing~\(u\) from her strategy. Otherwise, adopting~\(S_u'\setminus \autoset{v}\) is a better-response than~\(S_u'\), which contradicts our assumption.

The remaining case is that~\(S_u'\subset S_u\) holds, so the new strategy does not add any new target nodes but only removes existing ones. As shown by Corbo and Parkes~\cite{corbo2005price}, this implies that there must also be a single target node~\(v\in S_u\) whose removal is beneficial for agent~\(u\). Thus, the graph is not in~RE.
\end{proof}

\noindent\textbf{Pairwise Stability and Bilateral Greedy Equilibria:}
We compare graphs in PS and BGE with RE, BAE, and BSwE.
\begin{proposition}
PS is the intersection of the sets of BAE and RE. The set of GE is the intersection of the sets PS and BSwE.
\end{proposition}
\begin{proof}
The intersections follow directly from the definitions. Since both sets are visible in \Cref{fig:base_equilibria_venn}, we also know that BGE is a proper subset of PS.
\end{proof}

\noindent\textbf{Bilateral Neighborhood Equilibria:}
\begin{proposition}
The set of BNE is a proper subset of the intersection of the sets of
BAE and BGE.
\end{proposition}
\begin{proof}
Both subset relationships follow from the definition.
In \Cref{fig:example_graph_multi_add_greedy_not_neighborhood} we can see a graph, which is both in
BAE and in BGE but not in BNE.
\begin{figure}[h]
  \centering
  \includegraphics[width=0.3\textwidth]{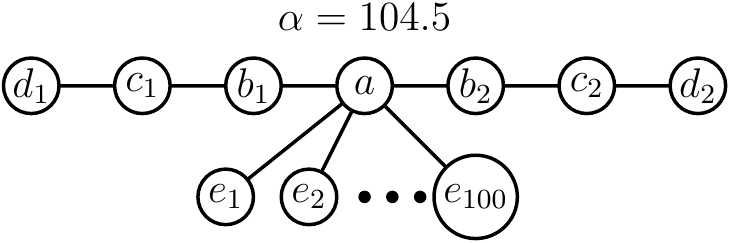}
\caption{Example of a graph that is in
BAE and in BGE but not in BNE.}
\label{fig:example_graph_multi_add_greedy_not_neighborhood}
\end{figure}

We start by arguing that it is in
BAE. Due to the high edge cost, agents only agree to pay for an additional edge, if their distance to node~\(a\) decreases. But one of the nodes involved in the construction of new edges has to be the closest one to node~\(a\), so it cannot be beneficial for all nodes involved.

For BGE, we already know that it is in AE and since the graph is a tree, it is also in RE, so it only remains to show that the graph is in BSwE.
Swapping around any of the type~\(c\),~\(d\) or~\(e\) nodes will not bring the swap partner closer to~\(a\), even though the partner has to pay for an extra edge, so it cannot be mutually beneficial. The type~\(b\) nodes already have ideal connections, so they also do not want to swap an incident edge. Finally, agent~\(a\) would benefit from swapping~\(ab_1\) for~\(ac_1\) (or~\(ab_2\) for~\(ac_2\)), but that change reduces the distance cost of~\(c_1\) only by~\(104\), which is less than the increase in buying cost.
Hence, the graph is in BGE.

On the other hand, the graph is not in BNE as~\(a\) can remove~\(ab_1\) and~\(ab_2\) while adding~\(ac_1\) and~\(ac_2\), so essentially doing two swaps simultaneously. This decreases the distance cost of~\(a\) by~\(2\). For~\(c_1\) and~\(c_2\), the distance cost is reduced by~\(105 > \alpha\), so they also benefit from the change.
\end{proof}

\noindent\textbf{Bilateral $k$-Strong Equilibrium:} We show that $k$-BSE and BNE are not comparable.
\begin{proposition}
There is a graph~\(G\) which is in BNE but not in~\(2\)-BSE. 
\end{proposition}

\begin{proof}
Consider graph~\(G\) in \Cref{fig:example_graph_neig_eq_not_2_eq}. We start by showing that it is in BNE.
\begin{figure}[h]
  \centering
  \includegraphics[width=0.2\textwidth]{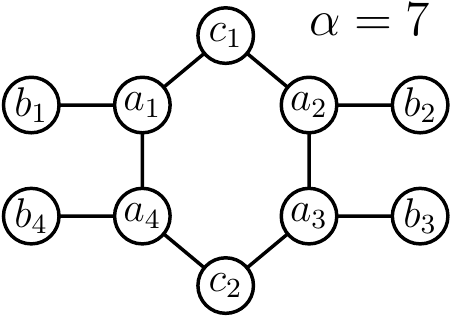}
\caption{This graph is a Neighborhood Equilibrium but not a 2-Strong Equilibrium.}
\label{fig:example_graph_neig_eq_not_2_eq}
\end{figure}
For brevity, we define the sets~\(A=\autoset{a_i}_{i\in[4]}\), \(B=\autoset{b_i}_{i\in[4]}\), and~\(C=\autoset{c_i}_{i\in[2]}\). For each of those sets, the nodes are symmetrical to all others in the set, which keeps the number of cases we need to consider small.

The initial distance costs are~\(\Dist{a_1} = 19\),~\(\Dist{b_1} = 27\), and~\(\Dist{c_1} = 19\). As there are~\(10\) nodes, the trivial distance cost lower bound for a node~\(u\in V\) is~\(2\cdot 9 - \card{S_v}\), this will be useful when limiting the number of additional edges the agents are willing to buy.

We start by showing that the nodes in~\(A\) and~\(C\) are not willing to cooperate by serving as new edge targets in a neighborhood change around some other node. To show this, we will utilize the following upper bound on the distance cost reduction.
Consider a neighborhood change around a node~\(u\in V\), where the change adds an edge to~\(v\in V\). Then, the distance cost reduction of~\(v\) is at most 
\[\sum_{w\in V}\maxset{0, \dist{v}{w}-2}+1.\] 
The reason for this is, if the distance to~\(w\in V\setminus\autoset{u}\) decreases by the change, then the new shortest path contains~\(u\), so the new distance is at least~\(2\). Only the distance to~\(u\) is decreased to~\(1\), which is the origin for the addition of~\(1\) in the term.

For~\(a\in A\), there are two nodes at distance~\(3\) and one at distance~\(4\), so the maximum distance reduction~\(a_1\) can hope to achieve by participating in a neighborhood change around a different node is~\(5\), which is less than the increase in buying cost for the extra edge. For~\(c\in C\), there are~\(3\) nodes at distance~\(3\), so the maximum distance cost reduction is even lower. Consequently, the nodes of~\(A\) and~\(C\) are not willing to act as new edge partners.

With that in mind, we now consider the different potential centers of the neighborhood changes.

First, consider a change around~\(a\in A\), without loss of generality let~\(a=a_1\). Let~\(S_a'\) denote the new strategy. Even if agent~\(a\) would add connections to all nodes in~\(B\), its distance cost would decrease by only~\(6\), which implies that~\(\card{S_c'}\leq 3\).
Based on the options we excluded, we can assume that~\(S_a' \subseteq B\cup\autoset{a_4, c_1}\) and~\(b_1\in S_a'\). Further, we observe that a connection to~\(a_4\) is never worse than a connection to~\(b_4\), and~\(c_1\) and~\(b_2\) are exchangeable. In the case~\(\card{S_a'}=1\), having a connection to~\(b_3\) is the best option, but that is still not profitable for agent~\(a\) as its distance cost would increase to~\(28\). For~\(\card{S_a'}=3\), the only interesting scenario is~\(b_3\in S_a'\), in which case the connection to~\(a_4\) is preferable over one to~\(c_1\) (which is exchangeable with~\(b_2\)). Here, so for~\(S_a' = \autoset{b_1, b_3, a_4}\), the distance cost increases to~\(21\), so agent~\(a\) also does not benefit.

Second, consider a change around~\(b\in B\), without loss of generality let~\(b=b_1\). Let~\(S_b'\) denote the new strategy. Even if it were to connect to all other nodes in~\(B\), its distance cost would decrease by only~\(12\), which implies~\(\card{S_b'}\leq 2\).
We can exclude~\(\card{S_b'} = 1\) because the target node~\(v\in S_b'\) only decreases its distance to~\(b\), which is not worth paying for an extra edge.
If~\(a_1\in S_b'\), the best target node for the additional edge is~\(b_3\), but that only decreases the distance cost by~\(6\) which is insufficient.
At the same time, a connection to~\(a_1\) is preferable over~\(b_4\), so we can also exclude~\(b_4\in S_b'\). This only leaves the case~\(S_b' = \autoset{b_2, b_3}\) which is clearly not an improvement for~\(b_1\).

Finally, consider a change around~\(c\in C\), without loss of generality let~\(c=c_1\). Let~\(S_c'\) denote the new strategy. Even if~\(c\) would add connections to all nodes in~\(B\), its distance cost would decrease by only~\(6\), which implies~\(\card{S_c'}\leq 2\).
With regard to the graph topology, a connection to~\(a_1\) is preferable over~\(b_1\) and~\(b_4\), and~\(a_2\) is preferable over~\(b_2\) and~\(b_3\). Hence, there is no improving change for~\(\card{S_c'}=2\). On the other hand, removing an existing edge increases the distance cost to~\(31\).

Thus, the graph~\(G\) is in BNE. By contrast, the coalition~\(\autoset{a_1, a_3}\) can improve by removing~\(a_1c_1\) and~\(a_3c_2\) while adding~\(a_1a_3\). Consequently, the graph~\(G\) is not in~\(2\)-BSE.
\end{proof}
Since BNE is a subset of BGE, we get the following corollary.
\begin{corollary}
 \(2\)-BSE are a proper subset of BGE.
\end{corollary}
Now we show that not all graphs in $k$-BSE are in BNE.
\begin{proposition}
For each~\(k\in\N_{\geq 2}\), there is a graph~\(G\) which is in~\(k\)-BSE but not in BNE. 
\end{proposition}
\begin{proof}
For a fixed~\(k\), consider the graph~\(G\) depicted in \Cref{fig:example_graph_k_se_not_neig_eq}. For brevity, we define the sets~\(B=\autoset{b_j}_{j\in[i]}\), \(C=\autoset{c_j}_{j\in[i]}\), and \(D=\autoset{d_j}_{j\in[i]}\).
\begin{figure}[h]
  \centering
  \includegraphics[width=0.2\textwidth,page=1]{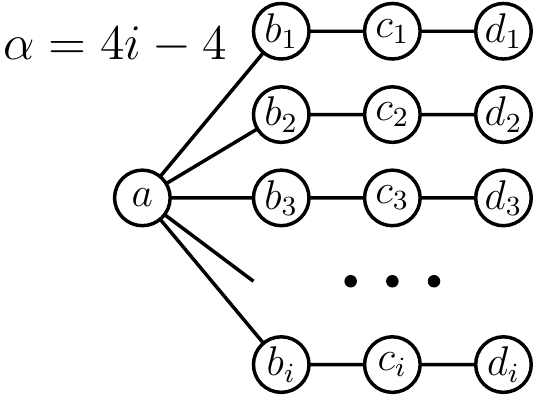}
\caption{For~\(k\in\N_{\geq 2}\) and~\(i=20k\), this graph is in~\(k\)-BSE but not in BNE.}
\label{fig:example_graph_k_se_not_neig_eq}
\end{figure}

We start by showing that~\(G\) is in~\(k\)-BSE.
Assume towards a contradiction a coalition~\(\Gamma\subset V\) of size at most~\(k\) which transforms the graph into~\(G'\) such that all members of the coalition reduce their cost.

For an agent~\(u\in V\), we account for two ways of reducing the distance cost. The first way is getting closer to~\(a\), where getting~\(x\in\N\) hops closer to~\(a\) reduces the distance cost by at most~\(xn\).
The other way of reducing the distance cost is based on getting closer to a row~\(R_j = \autoset{b_j, c_j, d_j}\) with~\(j\in[i]\) via a path not containing~\(a\) (unless~\(u=a\)). This is only viable if there is a node~\(v\in R_j\) which is part of the coalition~\(v\in\Gamma\). Seeing that the distance to~\(R_j\) in~\(G\) is at most~\(\dist{u}{R_j}\leq 15\), it can be reduced by at most~\(12\). So, the total distance reduction of this accounting category is upper bounded by~\(12k\).

With that in mind, the distance cost reduction for agent~\(b\in B\) is upper bounded by~\(12k\) and for agent~\(c\in C\) by~\(n+12k\), which implies that neither is willing to increase her buying cost.
On the other hand, agent~\(d\in D\) might be willing to pay for an extra edge if she gets a direct connection to agent~\(a\), but not otherwise.
This already limits the maximum degree of all nodes but node~\(a\) in~\(G'\) to~\(2\).
The agent~\(a\) cannot decrease the distance to herself, so we also know that agent~\(a\) is not willing to increase her buying cost.

If~\(a\notin\Gamma\), then no agent is willing to increase her degree. Since the number of edges must not decrease to ensure that~\(G'\) is connected, this even implies that each node has exactly the same degree in~\(G'\) as it did in~\(G\), so the agents can only hope to change their distance costs.
For all~\(b\in B\), this means they maintain their edges to node~\(a\), as increasing the distance to node~\(a\) and by extension all rows outside of the coalition outweighs any potential distance cost reduction within the coalition rows. The nodes~\(d\in D\) are the only one with degree~\(1\), so they remain the ends of each row. Finally, the nodes~\(c\in C\) insist on being connected to a node from~\(B\) to maintain their distance to node~\(a\), so they are also not willing to migrate to other rows beyond swapping places among each other (which evidently cannot be mutually beneficial). Consequently, there is no mutually beneficial coalitional move without agent~\(a\).

It remains to consider~\(a\in\Gamma\). As agent~\(a\) does not buy more edges than before and the degree of all other nodes is at most~\(2\) in~\(G'\), it is not possible for agent~\(a\) to decrease her distance cost since the paths attached to node~\(a\) are already perfectly balanced in size in~\(G\). This implies that agent~\(a\) needs to decrease her buying cost to benefit from the coalition.

Now, we conclude that the number of nodes with degree~\(1\) has to increase if agent~\(a\) decreases her degree. Each agent~\(b\in B\) which has her connection to node~\(a\) being removed, must have degree~\(1\) in~\(G'\) for the following reason. If~\(b\notin\Gamma\), she does not establish new connections, so losing the connection to node~\(a\) decreases the degree of node~\(b\). If~\(b\in\Gamma\), her distance cost increases from the change, so agent~\(b\) can only benefit from the change if her buying cost is reduced. On the other hand, the degree of~\(d\in D\) might be increased if agent~\(a\) builds an edge towards her, but that does not impact the number of nodes with degree~\(1\) as each new connection to a node~\(d\in D\) means that agent~\(a\) abandons an existing~\(b\in B\).

Finally, this means the number of nodes with degree~\(1\) is larger in~\(G'\) and therefore the number of nodes with degree~\(2\) must be smaller. At the same time, the degree of node~\(a\) decreases. Then, the graph~\(G'\) cannot be connected because it has less than~\(n-1\) edges, so the coalitional change is not beneficial.

In conclusion, this means that~\(G\) is in~\(k\)-BSE. On the other hand, there is a mutually beneficial neighborhood change around node~\(a\), in which all edges towards the nodes in~\(B\) are removed and edges towards all nodes in~\(C\) are created. This allows agent~\(a\) to reduce her distance cost while maintaining her buying cost. For agent~\(c\in C\), this change is beneficial because it decreases her distance cost from~\(4+12(i-1)\) to~\(3+8(i-1)\), so she improves by~\(1 + 4(i-1) > \alpha\).
\end{proof}

\noindent\textbf{Unilateral versus Bilateral Equilibria:}
We compare the NCG and the BNCG and with regard to Remove Equilibria and Add Equilibria in order to better understand their differences. In the unilateral NCG, the corresponding definition of an Add Equilibrium considers that a single agent might add a single edge without any strategy changes of other agents.

For simplicity, we assume for the NCG that each edge of the graph~\(G\) is owned by exactly one incident agent. This allows us to model the \emph{edge assignment} as a function~\(f:E\to V\), where each edge is mapped to one of its incident nodes. Under these assumptions, a graph~\(G\) and edge assignment~\(f\) completely capture the strategy vector of the NCG.

\thmone*
\begin{proof}
The first claim is shown by Corbo and Parkes~\cite[Proof of Proposition 5]{corbo2005price}.

For the other direction, consider the graph $G$ in~\Cref{fig:example_graph_add_eq_not_multi_add_eq}. 
\begin{figure}[h]
  \centering
  \includegraphics[width=0.2\textwidth]{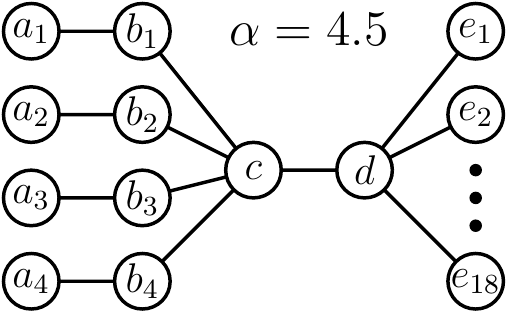}
\caption{An example of a graph that is in BAE but not in AE for the unilateral NCG.}
\label{fig:example_graph_add_eq_not_multi_add_eq}
\end{figure}
We argue that there are no single edges which can be added mutually beneficially. For brevity, we define the node subsets \(V_A=\autoset{a_i}_{i\in[4]}\), \(V_B=\autoset{b_i}_{i\in[4]}\) and \(V_E=\autoset{e_i}_{i\in[18]}\). Going through the different node types, we see that there is no mutually beneficial edge to be added.
\begin{itemize}
\item Node~\(d\) does not want to connect to anyone, as connecting to~\(a\in V_A\) or~\(b\in V_B\) only reduces its distance cost by~\(2\).
\item Node~\(c\) does not want to connect to anyone, as connecting to~\(a\in V_A\) or~\(e\in V_E\) only reduces its distance cost by~\(1\).
\item Node~\(b\in V_B\) does not want to connect to~\(e\in V_E\),~\(a\in V_A\) and~\(b'\in V_B\), as that reduces its distance cost by at most~\(2\).
\item Node~\(e\in V_E\) does not want to connect to~\(a\in V_A\) as that only reduces its distance cost by~\(4\). Connecting to~\(e\in V_E\) reduces the distance cost by~\(1\).
\item Node~\(a\in V_A\) does not want to connect with~\(a_j\) as it only reduces the distance cost by~\(4\).
\end{itemize}
Consequently,~\(G\) is in BAE.

However, agent~\(a_1\) benefits from buying~\(a_1d\), so $G$ with any edge assignment is not an Add Equilibrium in the unilateral NCG.
\end{proof}

\thmtwo*
\begin{proof}
Let~\(G\) be a Remove Equilibrium in the BNCG and consider an arbitrary edge assignment~\(f\). For any edge~\(uv\in E\), consider the owner~\(u=f(uv)\) in the assignment. Since~\(u\) does not improve by removing edge~\(uv\) in the bilateral game, she also does not improve by removing it in the unilateral game, as the resulting buying cost reduction and distance cost increase are the same in both games. Hence, the graph~\(G\) with edge assignment~\(f\) is in unilateral Remove Equilibrium.

For the other direction, let~\(G\) not be in Remove Equilibrium in the BNCG. Then there is an agent~\(u\in V\) and an edge~\(uv\in E\) such that agent~\(u\) improves by removing the edge~\(uv\).
Then, consider any edge assignment~\(f\) such that~\(f(uv)=u\). As agent~\(u\) owns~\(uv\) in this assignment, it has to pay for the edge~\(uv\) in the unilateral game and consequently benefits from removing it. Hence, the graph~\(G\) with edge assignment~\(f\) is not in Remove Equilibrium in the unilateral NCG.
\end{proof}

\lemmafournew*
\begin{proof}
We show that~\(C_n\) is in BSE for \(\frac{n^2}{4}-(n-1)<\alpha<\frac{n(n-2)}{4}\), if \(n\) is even, and for \(\frac{(n+1)(n-1)}{4} - (n-1) < \alpha < \frac{(n+1)(n-1)}{4}\), if~\(n\) is odd. We already know from \cite{corbo2005price} that~\(C_n\) is in RE for this range of~\(\alpha\).
Assume for contradiction that there exists a coalition~\(\Gamma\) that transforms the graph into~\(G'\) such that all members of the coalition benefit from this change.

First, consider the case where~\(G'\) has a degree of at most~\(2\). If every node in~\(G'\) has degree~\(2\), then~\(G'\) is isomorphic to~\(C_n\) and no agent benefits from the change.
If there are nodes with a lower degree, we conclude that there are exactly two nodes with degree~\(1\) and every other node has degree~\(2\), as the graph needs to be connected. Consequently, the graph~\(G'\) consists of a single path and one of its end nodes~\(u\in V\) must be part of the coalition. However, the change from agent~\(u\)'s perspective is equivalent to unilaterally removing an incident edge in~\(C_n\), which cannot be beneficial as~\(C_n\) is in RE. Thus, there is no improving coalitional move where~\(G'\) has degree~\(2\).

Thus, we can assume that there is a node~\(u\in V\) whose degree in~\(G'\) is larger than~\(2\). This agent~\(u\) must be part of the coalition~\(\Gamma\) and her buying cost increased by at least~\(\alpha\).
If~\(n\) is odd, we have \(\Dist[C_n]{u} = \frac{(n-1)(n-2)}{4}\). Combining this with the trivial lower bound \(\Dist[G']{u} \geq n-1\) implies that the distance cost of agent~\(u\) is reduced by at most \(\frac{(n-1)(n+1)}{4} - (n-1) < \alpha\).
Analogously, if~\(n\) is even, we have \(\Dist[C_n]{u} = \frac{n^2}{4}\) and the distance cost is reduced by at most \(\frac{n^2}{4}-(n-1) < \alpha\). So, in both cases, the distance cost reduction is smaller than the additional buying cost.

In conclusion, there cannot be a coalitional move where all members of the coalition strictly decrease their cost, so~\(C_n\) is in BSE for the given range of~\(\alpha\).
\end{proof}

\section{Omitted Details from Section~\ref{sec:preliminaries}}\label{app:prelim}

For the unilateral NCG, Albers, Eilts, Even-Dar, Mansour, and Roditty~\cite{albers2014nash} showed that the social cost of a NE can be upper bounded based on the distance cost of any node in the graph. We show a similar result for RE in the BNCG, albeit with slightly weaker bounds on the social cost.
\begin{restatable}{lemma}{lemmafour}
\label{lemma:cost_bound_by_dist_in_eqRemove}
Let~\(G\) be in RE and connected, then for any node~\(u\in V\) the social cost of $G$ can be upper bounded by
\[\Cost{G} \leq 2(n-1)(\alpha + \Dist{u}).\qedhere\]
\end{restatable}
\begin{proof}
Let~\(T\) denote a BFS-tree of~\(G\) from node~\(u\). We distinguish between tree edges~\(E_T\) and non-tree edges~\(E_G\setminus E_T\). Consider a node~\(v\in V\). Since RE and NE are equivalent in the BNCG, according to \Cref{theorem:remove_eq_is_nash_eq}, we consider the change where agent~\(v\) removes all incident non-tree edges. Her distance cost in the resulting graph would be at most~\(\Dist[T]{v}\), as all edges from~\(T\) would remain available.
This allows us to upper bound~\(\Cost[G]{v}\) by~\(\Cost[T]{v}\), because~\(G\) is in RE.
Applying this inequality to all nodes yields
\[\Cost{G} = \sum_{v\in V}\Cost[G]{v} \leq \sum_{v\in V}\Cost[T]{v} = \Cost{T}.\]

The proof concludes by upper bounding the social cost of~\(T\). To upper bound the total distance cost, we expand its definition \(\Dist{T} = \sum_{v\in V}\sum_{w\in V\setminus\autoset{v}}\dist[T]{v}{w}\) and then apply the triangle inequality to estimate~\(\dist[T]{v}{w} \leq \dist[T]{v}{u} + \dist[T]{u}{w}\).
Next, we rearrange the terms as follows
\begin{align*}
\Dist{T} &\leq \sum_{v\in V}\sum_{w\in V\setminus\autoset{v}}\left(\dist[T]{v}{u} + \dist[T]{u}{w}\right)\\
&= \sum_{v\in V}\Big((n-1)\dist[T]{v}{u} + \sum_{\mathclap{w\in V\setminus\autoset{v}}} \dist[T]{u}{w}\Big)\\
&= (n-1)\Dist[T]{u} + \sum_{v\in V}\sum_{w\in V\setminus\autoset{v}} \dist[T]{u}{w}\\
&= (n-1)\Dist[T]{u} + (n-1)\sum_{w\in V}\dist[T]{u}{w}\\
&= 2(n-1)\Dist[T]{u}.
\end{align*}
Finally, we add the buying cost for the~\(n-1\) edges in~\(T\) and get
\[\Cost{T} \leq 2(n-1)\alpha + 2(n-1)\Dist{u}.\qedhere\]
\end{proof}

In combination with the cost of the social optimum for~\(\alpha \geq 1\), \Cref{lemma:cost_bound_by_dist_in_eqRemove} allows us to upper bound the PoA by analyzing the distance cost of a single node. This is very powerful as this implies that it is sufficient to show a small distance cost for a single node. while the buying cost automatically follows suit.

\theoremfour*
\begin{proof}
For~\(\alpha\geq 1\), the social optimum~\(\opt\) has a total cost of~\(2(n-1)(\alpha+n-1)\).
Using \Cref{lemma:cost_bound_by_dist_in_eqRemove}, we upper bound \(\Cost{G}\) with \(2(n-1)(\alpha + \Dist{u})\).
By eliminating the common factor~\(2(n-1)\) from both sides of the division, we get
\[\poa{G} = \frac{\Cost{G}}{\Cost{OPT}} \leq \frac{\alpha + \Dist{u}}{\alpha + n-1}.\qedhere\]
\end{proof}
Fabrikant, Luthra, Maneva, Papadimitriou, and Shenker~\cite{fabrikant2003network} showed for the NCG that the maximum distance between two nodes in an Add Equilibrium is at most~\(2\sqrt{\alpha}\). The same bound applies in the BNCG. This and~\Cref{lemma:dist_u_bounds_poa_in_eqRemove} gives us an alternative proof for the PoA bound for PS shown by Corbo and Parkes~\cite{corbo2005price}.

\corollaryfive*
\begin{proof}
For any node~\(u\in V\), we can trivially derive~\(\Dist{u} < n^2\). Applying \Cref{lemma:dist_u_bounds_poa_in_eqRemove} yields
\[\poa{G} \leq \frac{\alpha + n^2}{\alpha + n-1} \leq 1 + \frac{n^2}{\alpha}. \qedhere\]
\end{proof}

\section{Omitted Details from Section~\ref{sec:swap_on_tree}}\label{app:swap_on_tree}

\lemmaseven*
\begin{proof}
We show this by induction over the cardinality of~\(T_u\). For~\(\card{T_u} = 1\), we have~\(\depth{T_u} = 0\) so the claim evidently holds.
Assume by induction that the claim holds for all subtrees of cardinality smaller than~\(\card{T_u}\). According to \Cref{lemma:close_median}, there is a~\(T_u\)-1-median~\(v\in T_u\) with~\(\layer{v} \leq \layer{u} + \frac{2\alpha}{n}\). By definition of the~\(T_u\)-1-median, for all~\(w\in T_u\) with~\(\layer{w} \geq \layer{v} + 1\) it holds that~\(\card{T_w} \leq \frac{\card{T_u}}{2}\) and, by induction, it follows that
\[\depth{T_w} \leq \left(1+\frac{2\alpha}{n}\right) \log \card{T_w} \leq \left(1+\frac{2\alpha}{n}\right) (\log \card{T_u} - 1).\] 
Thus, we can upper bound the depth of~\(T_u\) by
\[\depth{T_u} \leq 1 + \frac{2\alpha}{n} + \left(1+\frac{2\alpha}{n}\right) (\log \card{T_u} - 1) = \left(1+\frac{2\alpha}{n}\right) \log \card{T_u},\]
which concludes the induction and the proof.
\end{proof}

\lemmaeight*
\begin{proof}
If~\(\alpha \geq \frac{n}{2}\), the claim holds by definition of the 1-median~\(r\), so we assume~\(\alpha < \frac{n}{2}\) for the remainder of the proof.
Assume towards a contradiction that~\(\card{T_u} > \frac{\alpha}{\layer{u} - 1}\) and let~\(p\in V\) denote the parent of~\(u\). We argue that~\(u\) would swap the edge~\(up\) for~\(ur\) in this case.

Let~\(a\in V\) be the child of~\(r\) such that~\(u\in T_a\).
The swap is beneficial for~\(u\) as she gets closer by~\(\layer{u}-1\) to at least~\(\card{V\setminus S_a} \geq \frac{n}{2}\) nodes, while her distance to at most~\(\card{T_a\setminus T_u} < \frac{n}{2}\) nodes increases by at most~\(\layer{u}-1\).
For the root~\(r\), the distance cost decreases by~\(\card{T_u} (\layer{u} - 1)> \alpha\), which outweighs the increase in buying cost. Hence,~\(G\) is not in BSwE.
\end{proof}

\section{Omitted Details from Section~\ref{sec:greedy_on_tree}}\label{appendix:BGE}

\theoremten*
\begin{proof}
We already established that $2$-BSE is a subset of BGE, so it remains to show that every graph in BGE is also in $2$-BSE.

We show the contrapositive.
Assume~\(G\) is not in $2$-BSE. If~\(G\) is not in RE, it is also not in BGE. 
If~\(G\) is in RE, then there must be a coalition of two nodes~\(u,v\in V\) who can improve by building an edge between each other and potentially deleting incident edges. To maintain connectivity of the tree, they must not remove more edges than they add. Thus, the coalitional move is either adding a single edge or a swap around~\(u\) or~\(v\). Either way, this implies that~\(G\) is not in BGE.
\end{proof}

Now we introduce the \emph{stretched binary tree} that will be used for the PoA bound of~\(\Omega(\log\alpha)\).
A \emph{stretched binary tree}~\(T\) with parameters~\(d\in\N\) and~\(k\in\N_{\geq 1}\) is defined as follows.
Let~\(B\) be a complete binary tree of depth~\(d\) with root~\(r\). For~\(u\in V_B\setminus\autoset{r}\), we define~\(P_u = \autoset{u^i}_{i\in[k-1]}\cup\autoset{u}\) and define the node set of~\(T\) as~\(V_T = \autoset{r}\cup\bigcup_{u\in V_B\setminus\autoset{r}}P_u\). For~\(uv\in E_B\), where~\(u\) is the parent, the tree~\(T\) contains the edges~\(uv^1, v^1v^2,\dots,v^{k-1}v\). See \Cref{fig:stretched_binary_tree} for an example.

The resulting graph has~\(\left(\card{V_B}-1\right)k+1 = \left(2^{d+1}-2\right)k+1\) nodes and is a tree. For~\(u,v\in V_B\), it holds that~\(\dist[T]{u}{v} = k\cdot\dist[B]{u}{v}\) and thus~\(\depth{T} = k\cdot\depth{B}\).

The next lemma will be useful later, when computing a lower bound on the PoA.

\begin{lemma}
\label{lemma:average_depth_in_stretched_tree}
Let~\(T\) be a stretched binary tree with parameters \(d, k\in\N\), then the average layer of the nodes in~\(T\) is lower bounded by~\(k\left(d - \frac{3}{2}\right)\).
\end{lemma}
\begin{proof}
Let~\(B\) be the original binary tree of depth~\(d\).
We start by computing the sum of the layers~\(\sum_{u\in V_T}\layer[T]{v}\). As the root~\(r\) has layer~\(0\), we can ignore it. We partition the remaining nodes according to their paths, getting~\(\sum_{u\in V_B\setminus\autoset{r}}\sum_{v\in P_u}\layer[T]{v}\).

For a node~\(u\in V_B\setminus\autoset{r}\), there are~\(k\) nodes in~\(P_u\). The layers on the path range from \((\layer[B]{u} - 1)k + 1\) for~\(u^1\) up to~\(\layer[B]{u}k\) for~\(u\), so the average layer in~\(P_u\) is at least~\((\layer[B]{u}-\frac{1}{2})k\). Consequently, it holds that \[\sum_{v\in P_u}\layer[T]{v} \geq k\left(\layer[B]{u}-\frac{1}{2}\right)k.\] For~\(i\in[d]\), there are~\(2^i\) nodes in layer~\(i\) of~\(B\), so we get the inequality
\[\sum_{u\in V_T}\layer[T]{v} \geq k^2\sum_{i=1}^{d}2^i\left(i-\frac{1}{2}\right).\]
By deriving a closed form for the right hand side formula, we get
\[\sum_{u\in V_T}\layer[T]{v} \geq k^2(2^d(2d-3)+3) \geq k^22^{d}(2d-3).\]
From this, we now compute a lower bound for the average node layer, by dividing by the number of nodes~\((2^{d+1}-2)k+1 \leq k2^{d+1}\), yielding
\[\frac{k^22^{d}(2d-3)}{k2^{d+1}} \geq k\left(d - \frac{3}{2}\right).\qedhere\]
\end{proof}

Next, we show two lemmas to simplify the upcoming proofs for BAE and BSwE. These lemmas will be useful to reduce the number of cases we have to consider.
\begin{lemma}
\label{lemma:binary_tree_connection_level}
Let~\(T\) be a~\(k\)-stretched binary tree based on a complete binary tree~\(B\). Consider two nodes~\(u\in V_T\) and~\(v\in V_B\) such that~\(u\notin T_v\). Let~\(p,c\in T_v\) be nodes such that~\(p\) is the parent of~\(c\) and~\(\dist{c}{v} < \dist{u}{v}\). Then, agent~\(u\) prefers adding the edge~\(up\) over adding the edge~\(uc\). In particular, for~\(\dist{u}{v} > k\) holds that agent~\(u\) prefers connecting to~\(v\) over connecting to any other node in~\(T_v\).
\end{lemma}
\begin{proof}
We show that~\(\Dist[T+up]{u} < \Dist[T+uc]{u}\). For any node \(w\in V_T\setminus T_v\), it holds that~\(u\) is not closer to~\(w\) in~\(T+uc\) than it is in~\(T+up\). Consequently, it suffices to consider the nodes in~\(T_v\).

For all nodes~\(w\in T_c\), we have~\[\dist[T+up]{u}{w} = \dist[T+uc]{u}{w} + 1.\]
Since we assume that \(\dist{c}{v} < \dist{u}{v}\) and therefore also \(\dist{p}{v} + 1 < \dist{u}{v}\), we get that \[\dist[T+up]{u}{v} = \dist[T+uc]{u}{v} - 1\] and thus for \(w\in T_v\setminus\autoset{T_c}\) we get \[\dist[T+up]{u}{w} = \dist[T+uc]{u}{w} - 1.\] Then the claim holds because~\(\card{T_v\setminus\autoset{T_c}}\) is larger than~\(\card{T_c}\).

For the case \(\dist{u}{v} > k\), let \(c\in T_v\) be any node of the subtree, it is not restricted by \(\dist{c}{v}\geq \dist{u}{v}\). Let \(v'\in V_B\) denote the closest real ancestor of \(c\) included in \(B\). Then, we have \(\dist{c}{v'} \leq k\) and \(\dist{u}{v'} \geq \dist{u}{v} > k\), so agent~\(u\) prefers connecting to the parent of~\(c\). Hence, no node other than \(v\) can be the best option, which finishes the proof as there needs to be a best option.
\end{proof}

\begin{lemma}
\label{lemma:add_edge_across_subtrees}
Let~\(G\) be a tree with root \(r\) and let \(a, b\in \neig{=1}{r}\) be children of $r$ such that there is a layer-preserving isomorphism~\(f\) from~\(T_a\) to ~\(T_b\). If there are~\(u,v\in T_a\) with~\(\layer{u} \leq \layer{v}\), then~\(u\) prefers adding an edge to~\(v'=f(v)\) over adding an edge to~\(v\).
\end{lemma}
\begin{proof}
The distance reductions discussed in this proof always refer to~\(u\)'s distance cost. From \(\layer{u} \leq \layer{v}\) follows that the distance from~\(u\) to~\(r\) is not reduced by adding~\(uv\), so it only decreases the distance to nodes in~\(T_a\). Analogously, adding~\(uv'\) only decreases the distance to nodes in~\(T_b\).

Let~\(w\in T_a\) denote a node for which the distance improves by adding~\(uv\). This implies that the shortest~\(u\)-\(w\)-path in~\(T+uv\) contains~\(v\).
Then, since~\(f\) is an isomorphism, this implies for the node~\(w' = f(w)\in T_b\) that~\(\dist[G+uv']{u}{w'}=\dist[G+uv]{u}{w}\).
We upper bound~\(\dist[G]{u}{w}\)  by~\(\layer{u}+\layer{w}-2\) as~\(u\) and~\(v\) have the common ancestor~\(a\). By contrast, the only path between~\(u\) and~\(w'\) contains~\(r\), from which follows that~\(\dist[G]{u}{w'} = \layer{u} + \layer{w'}\).
Consequently, the distance reduction achieved towards~\(w'\) by adding~\(uv'\) is larger than the distance reduction towards~\(w\) by adding~\(uv\). And since we have a unique~\(w'\) for each~\(w\in T_a\), the total distance reduction must also be larger when adding~\(uv'\).
\end{proof}

Now we prove a sufficient condition under which a stretched binary tree is in BAE.
\begin{restatable}{lemma}{lemmaeleven}
\label{lemma:stretched_binary_tree_in_eqAdd}
Let~\(T\) be a~\(k\)-stretched binary tree and~\(\alpha\geq 5kn\). For~\(u,v \in V\) with~\(\layer{u} \leq \layer{v}\), agent~\(u\) does not benefit from adding~\(uv\). This implies that~\(T\) is in BAE.
\end{restatable}
\begin{proof}
We start by bounding~\(\dist{u}{v}\).
For any node~\(w\in V\), the distance from~\(u\) is only reduced by adding~\(uv\), if the shortest path in~\(T+uv\) contains~\(v\). This implies that \[\dist{u}{w} - \dist[T+uv]{u}{w} \leq \dist{u}{v}-1.\] Hence, for the total distance cost improvement to be greater than~\(\alpha\), it must hold that \[\dist{u}{v} > \frac{\alpha}{n} \geq 5k.\]

In the case~\(\layer{u}\leq k\), we can follow that \(\layer{v}> 4k\) based on the minimum distance requirement. Moreover, for~\(\layer{v} > 2k\), we know that~\(v\) is a worse candidate than its ancestor in layer~\(2k\), according to \Cref{lemma:binary_tree_connection_level}, and since we already excluded said ancestor, we can assume that~\(\layer{u} > k\) holds for the remainder of the proof.

Let~\(a,b\in V_B\) be the children of~\(r\) in~\(B\). By \Cref{lemma:add_edge_across_subtrees} and by symmetry, we can assume that~\(u\in T_a\) and~\(v\in T_b\). 
Furthermore, it suffices to consider the case~\(\layer{u} = \layer{v}\), according to \Cref{lemma:binary_tree_connection_level}.

Since~\(u\) does not get closer to~\(r\) by adding~\(uv\), she can only decrease her distance towards nodes in~\(T_b^1\). Let~\(v'\in V_B\cap T_b\) be the closest ancestor (including~\(v\)) of~\(v\) in~\(B\). We upper bound the distance reduction for~\(u\) when adding~\(uv'\) as it also limits the reduction for adding~\(uv\).

Let~\(i=\layer[B]{v'}\) and let~\(b=x_1, \dots, x_i=v'\) denote the path from~\(b\) to~\(v'\) in~\(B\). For~\(j\in[i-1]\), the distance to~\(x_j\) and hence to every node in~\(T_{x_j}^1\setminus T_{x_{j+1}}^1\) is reduced by at most
\begin{align*}
\dist{u}{x_j} - \dist[T+uv']{u}{x_j} &= \layer{u} + \layer{x_j} - \left(\layer{v'} - \layer{x_j} + 1\right)\\
 &= 2\layer{x_j} + \left(\layer{u} - \layer{v'}\right) - 1\\
 &\leq 2kj + k = 2k\left(j+\frac{1}{2}\right).
\end{align*}
With this, we can proceed to upper bound the distance benefit of agent~\(u\) as follows:
\begin{align*}
&\Dist[T]{u} - \Dist[T+uv']{u}\\
<{}& \sum_{j=1}^{i-1} 2k\left(j+\frac{1}{2}\right)\card{T_{x_j}^1\setminus T_{x_{j+1}}^1} + 2k\left(i+\frac{1}{2}\right)\card{T_{x_i}^1}\\
\leq{}& 2k\sum_{j=1}^{i} \left(j+\frac{1}{2}\right)\card{T_{x_j}^1}\\
\leq{}& 2k\sum_{j=1}^{i} \left(j+\frac{1}{2}\right)2^{-j}n\\
\leq{}& 2kn\sum_{j=1}^{\infty} \left(j+\frac{1}{2}\right)2^{-j}\\
={}& 5kn \leq \alpha.
\end{align*}
Finally, this means that the distance benefit is smaller than the cost~\(\alpha\) of adding the edge, so~\(u\) does not benefit from adding~\(uv'\) and, by extension, also not from adding~\(uv\). Consequently,~\(T\) is in BAE.
\end{proof}

The next lemmas will help us in the upcoming proof for BSwE.

\begin{lemma}
\label{lemma:binary_tree_swap_level}
Let~\(T\) be a~\(k\)-stretched binary tree with nodes~\(u,v\in V\) such that~\(u\) is the child of~\(v\). With~\(i\in\N\), let~\(x_1,\dots, x_i\in V_B\) be a downwards path in~\(B\) with~\(u\notin T_{x_1}\). Then, it holds that
\[\Dist[T - uv + ux_i]{u} \geq \Dist[T - uv + ux_1]{u} + k\card{T_{x_1}} (i-3).\qedhere\]
\end{lemma}
\begin{proof}
For~\(j\in[i]\), we define the shorthand~\(T(x_j)\) for~\(T - uv + ux_j\), which is the graph in which~\(u\) has swapped~\(uv\) for~\(ux_j\). 
By considering a swap as a deletion followed by an addition, we can apply \Cref{lemma:binary_tree_connection_level}.
It follows that \(\Dist[T(x_2)]{u} \geq \Dist[T(x_1)]{u}\). For~\(j\geq 3\), we proceed recursively with \[\Dist[T(j)]{u} \geq \Dist[T(j-1)]{u} + k\card{T_{x_1} \setminus T_{x_{j-1}}},\] as the distance to all nodes outside of~\(T_{x_{j-1}}\) increases by~\(k\) while the total distance within~\(T_{x_{j-1}}\) does not improve according to \Cref{lemma:binary_tree_connection_level}. By applying the lower bound~\(\card{T_{x_1} \setminus T_{x_{j-1}}} \geq \card{T_{x_1}}(1-2^{-(j-1)})\), we get
\begin{align*}
\Dist[T(x_i)]{u} &\geq \Dist[T(x_1)]{u} + k\card{T_{x_1}}\sum_{j=3}^i \left(1-2^{-(j-1)}\right)\\
&= \Dist[T(x_1)]{u} + k\card{T_{x_1}}\sum_{j=2}^{i-1} \left(1-2^{-j}\right)\\
&\geq \Dist[T(x_1)]{u} + k\card{T_{x_1}}\left(i-3\right).\qedhere
\end{align*}
\end{proof}

\begin{lemma}
\label{lemma:distance_to_path_after_swap_in_tree}
Let~\(G\) be a tree with~\(u,v,w\in V\) such that~\(uv\in E\) and~\(uw\notin E\). Let~\(G'\) denote the graph~\(G-uv+uw\).
With~\(i\in\N\), let~\(P = x_1, \dots, x_i\) denote a path such that the path from~\(u\) to each node in~\(P\) contains~\(v\) and~\(x_1\). If the path from~\(u\) to each node in~\(P\) contains~\(w\) and~\(x_i\) in~\(G'\), then it holds that
\[\dist[G']{u}{P} = \dist[G]{u}{P} + i(\dist{w}{x_i} - \dist{v}{x_1}).\qedhere\]
\end{lemma}
\begin{proof}
The idea of the proof is that all paths from~\(u\) to nodes in~\(P\) go through one of the endpoints of~\(P\)  in~\(G\) and~\(G'\), and by symmetry of the path, it does not matter through which endpoint they go.

We analyze~\(\dist[G]{u}{P}\) by breaking it down into multiple components, which yields
\begin{align*}
	\dist[G]{u}{P} &= \sum_{x\in P}(\dist[G]{u}{v} + \dist[G]{v}{x_1} + \dist[G]{x_1}{x}) \\
	&= i + i\cdot\dist[G]{v}{x_1} + \sum_{j=0}^{i-1}j.
\end{align*}
Analogously, it holds that~\[\dist[G']{u}{P} = i + i\cdot\dist[G']{w}{x_i} + \sum_{j=0}^{i-1}j.\] Computing the difference between both terms yields that~\[\dist[G']{u}{P} - \dist[G]{u}{P} = i\cdot\dist[G']{w}{x_i} - i\cdot\dist[G]{v}{x_1},\] which is equivalent to the claim.
\end{proof}

Next, we prove a sufficient condition under which a stretched binary tree is in BSwE.
\begin{restatable}{lemma}{lemmatwelve}
\label{lemma:stretched_binary_tree_in_eqSwap}
Let~\(T\) be a~\(k\)-stretched binary tree, then~\(T\) is in BSwE for~\(\alpha \geq 7kn\).
\end{restatable}
\begin{proof}
Let~\(B\) be the complete binary tree on which~\(T\) is based on.
We show that there are no nodes~\(u,v,w\in V\), with~\(uv\in E, uw\notin E\) such that swapping~\(uv\) for~\(uw\) is beneficial for agents~\(u\) and~\(w\). For~\(x\in V\), we define the shorthand~\(T(x)\) to denote~\(T-uv + ux\).

From the perspective of agent~\(w\), the swap is not better than simply adding~\(uw\), so, by~\Cref{lemma:stretched_binary_tree_in_eqAdd}, we can exclude that \(\layer{u} \leq \layer{w}\).
Moreover, it holds that \(\dist{u}{w} > \frac{\alpha}{n} \geq 7k\), as agent~\(w\) decreases her distance to at most~\(n-1\) nodes by at most~\(\dist{u}{w} - 1\) each when accepting the swap.

If~\(v\) is a child of~\(u\), then we need~\(w\in T_v\) to keep the graph connected. Due to the minimum distance between~\(u\) and~\(w\), there need to be at least six nodes from~\(B\) on the path from~\(u\) to~\(w\), not including~\(u\) and~\(w\) themselves. With~\(i\in\N_{\geq 6}\), let~\(x_1, \dots, x_i\in V_B\) denote these nodes ordered by proximity to~\(u\).

Now, we argue that agent~\(u\) does not benefit from swapping towards~\(w\) by upper bounding the distance cost reduction for~\(u\).
Note that it suffices to focus on the distance to the nodes in~\(T_{x_1}\), as, by~\Cref{lemma:distance_to_path_after_swap_in_tree}, the total distance to the nodes on the path between~\(v\) and~\(x_1\) will not decrease by the swap.

As \(\dist{u}{x_1} \leq k\), we can upper bound the benefit of swapping to~\(x_1\) by \(k\card{T_{x_1}}\), from which the inequality \[\Dist[T(x_1)]{u} \geq \Dist[T]{u} - k\card{T_{x_1}}\] follows.
Using \Cref{lemma:binary_tree_swap_level}, we get
\[\Dist[T(x_4)]{u} \geq \Dist[T]{u} - k\card{T_{x_1}} + k\card{T_{x_1}}(4-3) = \Dist[T]{u},\]
so agent~\(u\) does not benefit from swapping~\(uv\) for~\(ux_4\).
Moreover, by~\Cref{lemma:binary_tree_connection_level}, it holds that \(\Dist[T(w)]{u} \geq \Dist[T(x_4)]{u}\), so there cannot be a mutually beneficial swap. Hence, we assume for the rest of the proof that~\(u\) is a child of~\(v\).

Since~\(v\) is the parent of~\(u\), it must hold that \(w\notin T_u\) to keep the graph connected. 
Let~\(a, b\in B\) be children of the root~\(r\) in~\(B\). Since the distance and layer requirements allow us to exclude \(\layer{w} \leq \frac{7k}{2}\), we assume without loss of generality that~\(u\in T_a\) and~\(w\in T_b\).

Now, we establish necessary conditions for the swap to reduce the distance cost of~\(u\). Swapping does not get~\(u\) closer to~\(r\) or to any node in~\(V\setminus T_b^1\). By~\Cref{lemma:distance_to_path_after_swap_in_tree}, we can also neglect the nodes in~\(\neig{\leq k-1}{r}\setminus T_u\), as they form a path. Thus, it suffices to consider the distance cost reduction towards nodes in~\(T_b\).

We again consider a sequence of nodes from~\(B\). With~\(i=\floor{ \frac{\layer{w}}{k}}\), let~\(x_1, \dots, x_i\in V_B\) denote the nodes from~\(B\) on the path from~\(b\) to~\(w\). In particular, these nodes form a path in~\(B\) and~\(x_1\) denotes the node~\(b\). We upper bound the distance cost reduction by considering swaps to these intermediate nodes.
When swapping~\(uv\) for~\(ub\) and thus transforming~\(T\) to~\(T(b)\), the distance cost of agent~\(u\) restricted to~\(T_b\) is reduced by at most \[(\layer{u}+\layer{b} - 1)\card{T_b} < (\layer{u} + k)\card{T_b}).\] Applying~\Cref{lemma:binary_tree_swap_level}, we get that
\begin{align*}
\dist[T(x_i)]{u}{T_b} &\geq \dist[T]{u}{T_b} - (\layer{u} + k)\card{T_b} + k\card{T_b}(i-3) \\
&= \dist[T]{u}{T_b} - k\card{T_b}\left(\frac{\layer{u}}{k} + 4 - i\right).
\end{align*}
In order for \(\dist[T(x_i)]{u}{T_b}\) to be less than \(\dist[T]{u}{T_b}\), we need \(\frac{\layer{u}}{k} + 4 - i > 0\). Consequently, it must hold that~\(i\) is less than~\(\frac{\layer{u}}{k}+4\), which implies that
\[\layer{w} \leq \layer{x_i} + k-1 < k\left(\frac{\layer{u}}{k}+4\right) + k-1 < \layer{u} + 5k.\]

On the other side of the swap, agent~\(w\) only gets closer to nodes in~\(T_u\), so her distance reduction is \(\card{T_u}(\dist{u}{w} - 1)\). We insert \(\layer{w} < \layer{u}+5k\) to get \(\dist{u}{w} \leq 2\layer{u} + 5k\). Let \(l = \floor{\frac{\layer{u}}{k}}\), then we can upper bound \(\card{T_u}\) by~\(2^{-l}n\).
It follows that the distance cost reduction~\(w\) is at most
\[\Dist[T(w)]{w} - \Dist[T]{w} < 2^{-l}n(2(l+1)k + 5k).\]
The bound \(2(l+1)k + 5k \leq (2l + 7)k\) implies that \[\Dist[T(w)]{w} - \Dist[T]{w} < 2^{-l}(2l+7)nk.\] As this is maximized for~\(l=0\), we get that the distance reduction for agent~\(w\) is no more than~\(7kn\) and thus does not exceed the increase in buying cost. Finally, we conclude that there cannot be a swap beneficial for~\(u\) and~\(w\) at the same time. Thus, \(T\) is in BSwE.
\end{proof}

\theoremthirteen*
\begin{proof}
The graph~\(T\) is in RE because it is a tree, in BAE by \Cref{lemma:stretched_binary_tree_in_eqAdd}, and in BSwE by \Cref{lemma:stretched_binary_tree_in_eqSwap}.
\end{proof}

Before we proceed to the resulting lower bound proof for the PoA, we need a few bounds to show that we can build trees of a size close to a target node count~\(\eta\in\N\). Moreover, we establish multiple approximations which will be useful in the upcoming calculations.

\begin{lemma}
\label{lemma:stretched_tree_parameter_bounds}
For stretch factor~\(k\in\N_{\geq 1}\) and target size~\(t\in\R\), with~\(t\geq 2k+1\), let~\(T\) be a~\(k\)-stretched tree with~\(d\in\N\) chosen maximal subject to~\(n\leq t\). Then, it holds that \(\frac{t}{3} \leq n \leq t\) and \(k \log\left(\frac{t}{6k}\right) \leq \depth{T} \leq k\log t\).
\end{lemma}
\begin{proof}
The upper bound for~\(n\) holds trivially. Let \(L\subseteq V_B\) denote the leaves of the original tree~\(B\).
For all~\(v\in L\), the tree~\(T\) contains~\(P_v\), with~\(\card{P_v} = k\), since~\(t\geq 2k+1\). Let~\(A\) be the set of nodes added to~\(T\) if~\(B\) were to have an additional layer, then~\(\card{A}\) can be upper bounded by~\(2\card{\bigcup_{v\in L}P_v}\), which is less than~\(2n\). By definition of~\(T\), it holds that~\(n + \card{A} > t\), which implies that~\(n\geq \frac{t}{3}\).

The depth of~\(T\) is at most~\(k\log t\), as~\(\depth{B}\leq\log t\).
For the lower bound, we estimate~\(d\) for the original binary tree~\(B\). We already know that~\(n\geq\frac{t}{3}\) and~\(n_B\geq \frac{n}{k}\), so~\(B\) has at least~\(\frac{t}{3k}\) nodes. From this, we get 
\[d \geq \log\left(\frac{t}{3k}\right) -1 = \log\left(\frac{t}{6k}\right).\]
Multiplying by~\(k\) concludes the proof as~\(\depth{T} = k\depth{B}\).
\end{proof}

\theoremfourteen*
\begin{proof}
We choose~\(k=\ceil{\frac{\alpha}{\eta}} \leq 2\frac{\alpha}{\eta}\) and~\(d\) maximal such that~\(n\leq \frac{\eta}{14}\). Then, by \Cref{lemma:stretched_tree_parameter_bounds}, it holds that~\(\frac{\eta}{42} \leq n \leq \frac{\eta}{14}\). As~\(\alpha \geq 7nk\), the graph is in BGE according to~\Cref{theorem:stretched_tree_in_Greedy}.

It remains to compute~\(\poa{G}\). \Cref{lemma:stretched_tree_parameter_bounds} provides the following lower bound on the depth of the graph:
\begin{align*}
	\depth{B}&\geq \log\left(\frac{\frac{\eta}{14}}{6k}\right) \geq \log\left(\frac{\eta^2}{168\alpha}\right)\\
	&\geq \log\left(\eta^\gamma\right) - \log(168) > \gamma\log\eta - \frac{15}{2}.
\end{align*}
With this, we apply \Cref{lemma:average_depth_in_stretched_tree} to derive
\[\Dist{r} \geq n k\left(\depth{B} - \frac{3}{2}\right) \geq (n-1)\frac{\alpha}{\eta}\left(\gamma\log\eta - 9\right).\]
In order to approximate~\(\poa{G}\), we lower bound~\(\Dist{G}\) by~\(n\Dist{r}\). This yields
\[\poa{G} \geq \frac{2(n-1)\alpha + n(n-1)\frac{\alpha}{\eta}(\gamma \log\eta - 9)}{2(n-1)(\alpha + n -1)}.\]
Eliminating the factor~\((n-1)\) on both sides and lower bounding~\(n\frac{\alpha}{\eta}\) by~\(\frac{\alpha}{42}\) gives
\[\poa{G} \geq \frac{2\alpha + \frac{\alpha}{42}(\gamma \log\eta - 9)}{2(\alpha + n -1)} \geq \frac{\frac{25}{14}\alpha + \frac{\alpha}{42}\gamma\log\eta}{2(\alpha + n - 1)}.\]
With the inequalities~\(\alpha > \eta\) and~\(n\leq \frac{\eta}{7}\), we upper bound the denominator by~\(\frac{16}{7}\alpha\). This yields
\[\poa{G} \geq \frac{25\cdot 7}{14\cdot 16} + \frac{7}{42\cdot 16}\gamma\log\eta = \frac{25}{32} + \frac{1}{96}\gamma\log\eta.\qedhere\]
\end{proof}

Since we construct a lower bound for a large range of~\(\alpha\), we need to be able to scale our graph, so we can provide many different values of~\(n\) for a given~\(\alpha\). Thus, the next definition combines many stretched trees into a single graph.
We define a \emph{stretched tree star}~\(G\) with stretch factor~\(k\in\N_{\geq 1}\), target subtree size~\(t\in\R_{\geq 2k+1}\) and target size~\(\eta\in\N_{\geq 2t+1}\) as follows.
Let~\(T\) be a stretched tree with parameter~\(k\) and~\(d\) maximal subject to~\(|T|\leq t\). Then~\(G\) consists of a root~\(r\) with~\(\ceil{\frac{\eta-1}{|T|}}\) copies of~\(T\) as child subtrees.

The next lemma provides multiple bounds in preparation for the following proofs. It helps reduce redundancy by allowing calculations to be reused.

\begin{lemma}
\label{lemma:stretched_tree_star_parameter_bounds}
Let~\(G\) be a stretched tree star with \(k,t,\eta\in\N\). Then it holds that \(\eta \leq n \leq \frac{3}{2}\eta\) and \(\depth{T} \leq \depth{G} \leq 2k\log t\).
\end{lemma}
\begin{proof}
For the bounds on~\(n\), we consider that~\(n = \ceil{\frac{\eta-1}{|T|}}\ |T| + 1\). This is guaranteed to be at least~\(\eta\) for any positive~\(|T|\). For the upper bound, we consider that~\(\eta - n < |T|\), which is less than~\(\frac{\eta}{2}\) by the parameter restrictions of the definition. 

The depth of~\(G\) is equal to~\(\depth{T}+1\). By \Cref{lemma:stretched_tree_parameter_bounds}, the depth of~\(T\) is at most~\(k\log t\), and since this is at least~\(1\), we can upper bound~\(\depth{T} + 1\) by~\(2\depth{T}\).
\end{proof}

Next, we lower bound~\(\poa{G}\) based on the parameters.

\begin{restatable}{lemma}{lemmabeforefifteen}
\label{lemma:stretched_tree_star_poa}
Let~\(G\) be a stretched tree star with~\(k,t,\eta\in\N\), then 
\[\poa{G} \geq \frac{nk\left(\log\left(\frac{t}{k}\right) - \frac{9}{2}\right)}{2(\alpha + n - 1)}.\qedhere\]
\end{restatable}
\begin{proof}
Based on \Cref{lemma:stretched_tree_star_parameter_bounds}, we can lower bound the depth~\(d\) of the original binary tree~\(B\) with~\(d \geq \log\left(\frac{t}{6k}\right)\). Combining this with \Cref{lemma:average_depth_in_stretched_tree} yields
\[\Dist{r} \geq (n-1)\left(1 + k\left(d - \frac{3}{2}\right)\right) \geq (n-1)k\left(\log\left(\frac{t}{6k}\right) - \frac{3}{2}\right).\]
By lower bounding~\(\Cost{G}\) by~\(n\cdot\Dist{r}\), we follow that
\begin{align*}
		\poa{G} &\geq \frac{n\cdot\Dist{r}}{2(n-1)(\alpha + n-1)}\\
		&\geq \frac{n(n-1)k\left(\log\left(\frac{t}{6k}\right) - \frac{3}{2}\right)}{2(n-1)(\alpha + n-1)}\\
		&\geq \frac{nk\left(\log\left(\frac{t}{k}\right) - \frac{9}{2}\right)}{2(\alpha + n - 1)}.\qedhere
\end{align*}
\end{proof}

\section{Omitted Details from Section~\ref{sec:neighborhood_on_tree}}\label{app:neighborhood_on_tree}

\lemmasixteen*
\begin{proof}
We start by showing that the conditions imply that \(\alpha > 6\card{T}\cdot\depth{G}\) holds. Assume towards a contradiction that \(\alpha \leq 6\card{T}\cdot\depth{G}\), then we use \(n\geq 2\card{T}+1\) to get the lower bound \(\frac{3n\cdot\depth{G}}{\alpha}+1 > 2\) and the upper bound \(\frac{\alpha}{3\card{T}\cdot\depth{G}} \leq 2\), which contradicts the inequality in the premise.

Now, we consider a neighborhood change around a node~\(u\in V\) where edges to~\(A\subseteq V\) get added and existing edges to~\(R\subseteq V\) get removed. To maintain connectivity, it must hold that~\(\card{A} \geq \card{R}\).

First, we consider the case~\(u=r\). For each child~\(v\in V\) of~\(u\), there needs to be at least one edge between~\(u\) and~\(T_v\). Paying for more than one edge towards~\(T_v\) cannot benefit agent~\(u\) as \(\dist{u}{T_v}\leq \card{T}\cdot\depth{G} < \frac{\alpha}{6}\). Moreover, agent~\(u\) is already connected to the 1-median of~\(T_v\), so there is also no benefit in swapping the edge for a different target in~\(T_v\). Hence, agent~\(u\) cannot improve if she is the root.

Now, we argue that even if agent~\(u\) is not the root, she cannot improve by only applying changes restricted to~\(T_u\). Just as in the previous case, for each child~\(v\in V\) of~\(u\), there needs be at least one edge from~\(u\) to~\(T_v\), so the buying cost cannot be reduced. Again, buying an extra edge is also not beneficial since~\(\dist{u}{T_v} < \frac{\alpha}{6}\).

It remains to consider a swap from~\(uv\) to~\(uw\) with~\(w\in T_v\). If~\(k=1\), node~\(v\) is already the 1-median of~\(T_v\), so agent~\(u\) does not benefit from such a swap.
If~\(k>1\), agent~\(u\) might benefit from the swap.
In the proof of~\Cref{lemma:stretched_binary_tree_in_eqSwap}, we have already shown that such a swap can only reduce the distance cost of agent~\(u\) if it holds that~\(\dist{u}{w} < 5k\).
As the swap improves the distance cost of agent~\(w\) by at most \((\dist{u}{w}-1)n\), we upper bound the distance reduction for agent~\(w\) by \(5kn\leq\frac{5}{6}\alpha\), which is less than she has to pay for the new edge. Thus, agent~\(w\) is not willing to accept the swap without additional changes. And even if~\(u\) performs other changes at the same time, as long as they are restricted to~\(T_u\), the distance cost of agent~\(w\) will be reduced by no more than an additional~\(\dist{u}{T_u}\), which is less than~\(\frac{\alpha}{6}\).
Consequently, agent~\(u\) cannot improve by only doing changes within~\(T_u\).

So, assume for the rest of the proof that agent~\(u\) also does changes outside of~\(T_u\).
As agent~\(u\) might still combine these changes with changes internal to~\(T_u\), we account for an internal distance benefit of less than~\(\dist{u}{T_u} < \frac{\alpha}{6}\).
Next, our analysis focuses on the changes outside of~\(T_u\), so we define~\(A'=A\setminus T_u\) and~\(R'=R\setminus T_u\), where we assume that~\(A'\neq\emptyset\) and~\(\card{A'}\geq\card{R'}\). Among the nodes in~\(A'\), let \(v\in A'\) denote one with the smallest distance to~\(r\). In particular, it holds for all \(w\in A'\) that \(\layer{v} \leq \layer{w}\).

If \(\layer{v}\geq\layer{u}-1\), agent~\(u\) does not get any closer to~\(r\) by this change. The distance reduction provided by each new partner in~\(A'\) is at most \(2\cdot\depth{G}\card{T} < \frac{\alpha}{3}\).
Hence, the distance cost of agent~\(u\) is reduced by at most \(\card{A'}\frac{\alpha}{3}+\frac{\alpha}{6}\).
If \(\card{A'} > \card{R'}\), this implies that the distance cost reduction cannot be larger than the additional buying cost of \((\card{A'}-\card{R'})\alpha\) with~\(\card{R'}\leq 1\).
If \(\card{A'} = \card{R'} = 1\), the buying cost of agent~\(u\) does not increase, but as agent~\(u\) has removed her parent edge, the distance reduction for agent~\(v\) is upper bounded by \(2\cdot\depth{G}\card{T_u} < \frac{\alpha}{3}\), so agent~\(v\) does not benefit from the change.

Consequently, we assume that \(\layer{v} < \layer{u}-1\). We now consider requirements such that agents~\(u\) and~\(v\) both benefit from this change.
Then, we show that these requirements contradict each other from which follows that no mutually beneficial change exists.

The distance reduction for agent~\(u\) from the complete neighborhood change is trivially upper bounded by \(\Dist{u} \leq 2\cdot\depth{G}n\). Consequently, the number of new edges going outside of~\(T_u\) is upper bounded by
\[\card{A'} < \frac{2n\cdot\depth{G}}{\alpha} + 1,\]
where the addition of~\(1\) is due to the fact that~\(u\) might delete its parent edge, so~\(\card{R'}=1\).

For~\(v\), the buying cost increases by~\(\alpha\). Let \(C\subset V\) denote the children of~\(r\). We define \[C' = \autoset{c\in C \mid T_c\cap ((A'\cup\autoset{u})\setminus\autoset{v}) \neq \emptyset}\] and find that the distance from~\(v\) only decreases to descendants of the nodes in~\(C'\), as agent~\(v\) does not get any closer to the root~\(r\). For each~\(c\in C'\), the distance reduction to the nodes in~\(T_c\) is trivially upper bounded by \(\dist{v}{T_c} < 2\card{T}\cdot\depth{G}\). So, for~\(v\) to benefit from the change, it is necessary that \(\card{C'} > \frac{\alpha}{2\card{T}\cdot\depth{G}}\).
Since~\(u\notin A'\) and~\(v\in A'\), it holds that \(\card{A'}\geq \card{C'}\), and thus
\[\card{A'} > \frac{\alpha}{2\card{T}\cdot\depth{G}}.\]

Finally, we can combine the upper and lower bound to obtain
\[\frac{\alpha}{2\card{T}\cdot\depth{G}} < \card{A'} < \frac{2n\cdot\depth{G}}{\alpha} + 1.\]
By the assumption of the lemma, the lower bound is larger than the upper bound. Thus, there is no~\(\card{A'}\) fulfilling the requirements. Thus, no mutually beneficial neighborhood change exists.
\end{proof}

\theoremseventeen*
\begin{proof}
We prove both statements separately:
\begin{itemize}
 \item[(i)] We construct~\(G\) as a stretched tree star with parameters \(k=\floor{\frac{\alpha}{9\eta}}\),~\(t = \eta^{1-\varepsilon/2}\) and~\(\eta\) as provided.
Using \Cref{lemma:stretched_tree_star_parameter_bounds}, we upper bound~\(\depth{G}\) by~\(2k\log t\), which is upper bounded by~\(2\frac{\alpha}{9\eta}\log \eta\). We bound the terms from \Cref{lemma:stretched_tree_star_neighborhood} with
\begin{align*}
	&\frac{3n\cdot\depth{G}}{\alpha} + 1 \leq \frac{3n\cdot 2\frac{\alpha}{9\eta}\log \eta}{\alpha}+1\\
	={}&\frac{2n}{3\eta}\log\eta + 1 \leq \log \eta + 1,
\end{align*}
and with
\[\frac{\alpha}{3\card{T}\cdot\depth{G}} \geq \frac{\alpha}{3t\cdot 2\frac{\alpha}{9\eta}\log \eta} \geq \frac{3\eta}{2t\log \eta} = \frac{3\eta^{\varepsilon/2}}{2\log \eta}.\]
For sufficiently large~\(\eta\), it holds that~\(\log \eta+1\) is smaller than~\(\frac{3\eta^{\varepsilon/2}}{2\log \eta}\). Moreover, it holds that~\(6kn \leq \frac{2\alpha}{3\eta}n \leq \alpha\). Thus, we can apply \Cref{lemma:stretched_tree_star_neighborhood} to conclude that~\(G\) is in BNE.

It remains to show the logarithmic lower bound on~\(\poa{G}\). \Cref{lemma:stretched_tree_star_poa} gives us the lower bound
\[\poa{G} \geq \frac{nk\left(\log\left(\frac{t}{k}\right) - \frac{9}{2}\right)}{2(\alpha + n - 1)}.\]
Due to~\(n\leq\frac{3}{2}\eta\leq\frac{1}{6}\alpha\), we can upper bound the denominator by~\(\frac{7}{3}\alpha\).
In the numerator, we lower bound~\(nk\) by~\(n\frac{\alpha}{18\eta}\geq \frac{\alpha}{18}\) and~\(\frac{t}{k}\) by~\(9\eta^{\varepsilon/2}\).
Hence, we get
\begin{align*}
	\poa{G} &\geq \frac{\frac{\alpha}{18}\left(\log\left(9\eta^{\varepsilon/2}\right) - \frac{9}{2}\right)}{\frac{7}{3}\alpha}\\
	&= \frac{1}{42}\log\left(9\eta^{\varepsilon/2}\right) - \frac{3}{28}\\
	&\geq \frac{\varepsilon}{84}\log(\eta) - \frac{3}{28}.
\end{align*}
We conclude the proof by inserting~\(\alpha^{1/2} \leq \eta\), which implies that
\[\poa{G} \geq \frac{\varepsilon}{168}\log \alpha - \frac{3}{28}.\]
\item[(ii)] We construct~\(G\) as a stretched tree star with~\(k=1\),~\(t = \eta^\varepsilon\) and~\(\eta\) as provided.
Using \Cref{lemma:stretched_tree_star_parameter_bounds}, we can now bound \(\depth{G} \leq 2k\log t\), which is upper bounded by~\(2\varepsilon\log \eta\). We bound the terms from \Cref{lemma:stretched_tree_star_neighborhood} with
\[\frac{3n\cdot\depth{G}}{\alpha} + 1 \leq \frac{3n\cdot 2\varepsilon\log \eta}{\alpha}+1 \leq 9\varepsilon \eta^{1/2-\varepsilon}\log \eta + 1,\]
and with
\[\frac{\alpha}{3\card{T}\cdot\depth{G}} \geq \frac{\alpha}{3t\cdot 2\varepsilon\log \eta} \geq \frac{\eta^{1/2}}{6\varepsilon\log \eta}.\]
For sufficiently large~\(\eta\), it holds that~\(9\varepsilon \eta^{1/2-\varepsilon}\log \eta + 1\) is smaller than~\(\frac{\eta^{1/2}}{6\varepsilon\log\eta}\), so~\(G\) is in BNE according to \Cref{lemma:stretched_tree_star_neighborhood}.

It remains to show the logarithmic lower bound on~\(\poa{G}\). \Cref{lemma:stretched_tree_star_poa} gives us the lower bound
\[\poa{G} \geq \frac{nk\left(\log\left(\frac{t}{k}\right) - \frac{9}{2}\right)}{2(\alpha + n - 1)}.\]
Due to~\(\alpha\leq \eta\) and~\(\eta\leq n\), we can upper bound the denominator by~\(4n\). Inserting the parameter values in the numerator concludes the proof with the inequality
\[\poa{G} \geq \frac{n\left(\log\left(\eta^{\varepsilon}\right) - \frac{9}{2}\right)}{4n} = \frac{1}{4}\varepsilon\log \eta - \frac{9}{8} \geq \frac{1}{4}\varepsilon\log\alpha - \frac{9}{8}.\qedhere\]
\end{itemize}
\end{proof}

\theoremnineteen*
\begin{proof}
Let~\(r\) denote the 1-median and root of \(G\).
Let~\(u\in V\) be a node of maximum layer. If \(\layer{u} \leq 2\), the claim holds. So, we assume that \(\layer{u} \geq 3\). Moreover, it holds that \(\alpha < \frac{n}{2}\).
We assume that \(n > 4\) and thus \(\alpha < \frac{n}{2}\), as otherwise the diameter is at most~\(3\) and the claim holds.

With \(i = \card{\neig{=3}{r}}\), let \(\autoset{c_j}_{j\in[i]}\) denote the nodes in layer~\(3\) sorted descendingly by their subtree size. Let \(k=\minset{\floor{\frac{n}{\alpha}-1}, i}\).
We consider the following change around \(u\). Agent~\(u\) buys an edge towards \(r\) and towards the nodes in \(\autoset{c_j}_{j\in[k]}\setminus\autoset{u}\).

Agent~\(u\) pays for at most \(\frac{n}{\alpha}\) additional edges, so her buying cost increases by at most \(n\). Connecting to \(r\) decreases her distance cost by at least \(2\frac{n}{2}\). Further, agent~\(u\) profits from the new direct connections to the nodes in \(\autoset{c_j}_{j\in[k]}\setminus\autoset{u}\), so the overall change is beneficial for~\(u\). Each agent \(c\in\autoset{c_j}_{j\in[k]}\) profits because her distance to \(r\) decreases by~\(1\), so her distance cost decreases by at least \(\frac{n}{2}\) while she only has to pay for one additional edge.
Then, agent~\(r\) must not benefit from the proposed change, as \(G\) is in BNE.
As agent~\(r\) decreases her distance to each node in \(c\in\autoset{c_j}_{j\in[k]}\setminus\autoset{u}\) by~\(1\) and to \(u\) by \(2\), it must hold that
\(\sum_{j=1}^k \card{T_{c_j}} < \alpha.\)
Then, if \(i > k\), this allows us to upper bound \(\card{T_{c_{k+1}}}\) as follows
\[\card{T_{c_{k+1}}} < \frac{\alpha}{k} < \frac{\alpha}{\frac{n}{\alpha}-2} \leq \frac{\alpha}{\frac{n}{2\alpha}} = 2.\]
So, if the node~\(c_{k+1}\) exists, then it does not have any children.
This means that only the agents \(c\in\autoset{c_j}_{j\in[k]}\) have descendants beyond layer~\(3\). Carrying over from BGE, we get for \(c\in\autoset{c_j}_{j\in[k]}\) that \(\depth{T_c} \leq \log\alpha\).
We conclude that \(\Dist{r} \leq 3(n-1) + \alpha\log\alpha\) and apply \Cref{lemma:dist_u_bounds_poa_in_eqRemove} to get
\[\poa{G} \leq \frac{\alpha + 3(n-1) + \alpha\log \alpha}{\alpha + (n-1)} \leq 3 + \frac{\sqrt{n}\left(1 + \log\sqrt{n}\right)}{n} \leq 4.\qedhere\]
\end{proof}

\section{Omitted Details from Section~\ref{chapter:general_equilibria}}\label{appendix:general}

\theoremtwentytwo*
\begin{proof}
We start by lower bounding the cost for an agent \(u\in V\) depending on the number of incident edges denoted by \(\card{E_u}\). This will then allow us to argue that agents cannot benefit from any coalitional moves as they already have the lowest cost possible.
For any agent \(v\in V\), it holds that \(\dist{u}{v} = 1\) if the edge~\(uv\) is part of the graph and \(\dist{u}{v}\geq 2\) otherwise. Thus, the distance cost of agent~\(u\) is at least \(2(n-1) - \card{E_u}\). Adding the buying cost yields the following inequality for the total cost of agent~\(u\):
\[\Cost{u} \geq \card{E_u}\alpha + 2(n-1) - \card{E_u} = \card{E_u}(\alpha - 1) + 2(n-1).\]
For any connected graph, it holds that \(1\leq \card{E_u}\leq n-1\).

For \(\alpha < 1\), the lower bound is minimized when \(\card{E_u}\) is maximized. Namely, the bound is \((\alpha+1)(n-1)\). All agents of the clique already meet the lower bound, hence they cannot improve their cost any further. Furthermore, the clique is the only graph in BSE because Corbo and Parkes~\cite{corbo2005price} have already proven it to be the only pairwise stable graph.

If \(\alpha=1\), the first summand of the lower bound is zero, so only \(2(n-1)\) remains. If a graph~\(G\) has a diameter of at most \(2\), all agents meet that lower bound and thus cannot improve via any change. If \(G\) has a larger diameter, so if there are agents \(u,v\in V\) with \(\dist{u}{v} \geq 3\), they can improve by forming a coalition of size \(2\) and adding the edge~\(uv\).

For \(\alpha > 1\), the lower bound is minimal for \(\card{E_u} = 1\) and equals \(\alpha + 1 + 2(n-2)\). If~\(G\) is a star, this lower bound is met exactly by every leaf agent, so none of them want to change the graph. Only the center agent remains, and while her cost can be very large, she has no coalition partners and can only delete existing edges, which is not beneficial as that would disconnect the graph.

Other graphs in BSE exist for \(\alpha > 1\). For instance, a path of \(4\) nodes is in BSE for \(\alpha = 100\) since no agent is willing to increase her buying cost, as the additional edge cost outweighs any potential distance reduction.
\end{proof}

\lemmatwentythree*
\begin{proof}
We show that~\(\Cost{H}\) is at most \(2(n-1)\Cost[G]{u}\), then the claim follows directly by dividing by the socially optimal cost \(2(n-1)(\alpha + n - 1)\).
Let~\(v\in V_H\) be the agent with the lowest cost in~\(H\). If \(\Dist[H]{v} \leq \Cost[G]{u} - \alpha\), we apply \Cref{lemma:cost_bound_by_dist_in_eqRemove} and get
\[\Cost{H} \leq 2\alpha(n-1) + 2(n-1)(\Cost[G]{u} - \alpha) = 2(n-1)\Cost[G]{u},\]
so the claim holds.

Otherwise, we have \(\Dist[H]{v} > \Cost[G]{u} - \alpha\). Since~\(H\) is connected, the agent~\(v\) has to pay for at least one edge. Then it holds that \(\Cost[H]{v} > \Cost[G]{u}\), which implies that the worst-off agent in~\(G\) has a lower cost than the best-off agent in~\(H\). Hence, all agents in~\(H\) can form a coalition to transform into a graph that is isomorphic to~\(G\). As all agents benefit from this change, graph~\(H\) is not in BSE.
\end{proof}

\lemmatwentyfour*
\begin{proof}
No agent pays for more than~\(d\) child edges and one parent edge, so the buying cost of any agent is upper bounded by~\((d+1)\alpha\).
The depth of the tree is upper bounded by~\(\log_d n\), so the distance between any two nodes is at most~\(2\log_d n\), giving us the upper bound on the distance cost.
\end{proof}

\theoremtwentynine*
\begin{proof}
Slightly relaxed, the requirement means that no agent may have costs exceeding~\(2pn\). With regard to the buying costs, this implies that the degree must not be larger than~\(2p\), which is a constant.
We define the constant~\(k = \sum_{i=0}^{4p}(2p)^i \in\Theta(1)\), and observe that for no agent~\(u\) there may be more than~\(k\) nodes in~\(\neig{\leq 4p}{u}\), as each layer in the shortest path tree from~\(u\) can contain at most~\(2p\) as many nodes as the previous layer.

We now choose~\(n\) such that~\(\frac{n}{2}\) is larger than~\(k\). Let~\(G\) be a graph according to the premise and let~\(u\) be an agent. As the distance cost of agent~\(u\) is upper bounded by~\(2pn\), at least half of all nodes must be included in~\(\neig{\leq 4p}{u}\), but that contradicts our previous assessment that~\(\card{\neig{\leq 4p}{u}}\leq k < \frac{n}{2}\). Consequently, it is not possible for all agents in~\(G\) to have costs below~\(2pn\). In fact, the existence of an agent with distance cost below~\(2pn\) requires at least one node with a degree in~\(\Omega(\sqrt[4p]{n})\).
\end{proof}

\end{document}